\documentclass[aps,prx,reprint,groupedaddress,twoside,showpacs]{revtex4-2}

\usepackage{amsmath,amssymb}
\usepackage{amsthm}
\usepackage{hyperref} 
\usepackage{acronym}  
\usepackage{color}  
\usepackage{amsfonts}
\usepackage{mathrsfs}
\usepackage{graphicx}
\usepackage[on]{psfrag}
\usepackage{array}
\usepackage{cases}
\usepackage{tabularx}
\usepackage{url}
\usepackage{soul}
\usepackage{theoremref}
\usepackage{thmtools}
\usepackage{lettrine}
\usepackage{mathtools}
\usepackage{dsfont}
\usepackage[table]{xcolor}
\usepackage{multirow}
\usepackage{dashbox}
\usepackage{algorithm}
\usepackage{algpseudocode}
\usepackage{float}
\usepackage{enumitem}
\usepackage{hhline}
\usepackage{empheq}
\usepackage{dashrule}
\usepackage{ulem}
\newcommand\reddashuline{\bgroup\markoverwith
{\textcolor{red}{
\hdashrule[-0.5ex]{1pt}{0.8pt}{0.5pt}}}\ULon}
\newcommand\blueuline{\bgroup\markoverwith
{\textcolor{blue}{\rule[-0.5ex]{2pt}{0.8pt}}}\ULon}

\usepackage{tikz}

\usepackage{xpatch}

\makeatletter
\renewcommand{\ALG@name}{Protocol}
\makeatother

\usepackage{relsize}%
\usepackage{WinsNotation}

\makeatletter
\renewcommand*\env@matrix[1][\arraystretch]{%
  \edef\arraystretch{#1}%
  \hskip -\arraycolsep
  \let\@ifnextchar\new@ifnextchar
  \array{*\c@MaxMatrixCols c}}
\makeatother

\makeatletter
\newcommand{\multiline}[1]{%
  \begin{tabularx}{\dimexpr\linewidth-\ALG@thistlm}[t]{@{}X@{}}
    #1
  \end{tabularx}
}
\makeatother

\newcommand{\SM}{\hspace{-0.15mm}\Set M}
\newcommand{\tSM}{\hspace{-0.15mm}\tilde{\Set M}}

\makeatletter
\renewcommand*\env@matrix[1][\arraystretch]{%
  \edef\arraystretch{#1}%
 \hskip -\arraycolsep
  \let\@ifnextchar\new@ifnextchar
  \array{*\c@MaxMatrixCols c}}
\makeatother

\makeatletter
    \def\CT@@do@color{%
      \global\let\CT@do@color\relax
            \@tempdima\wd\z@
            \advance\@tempdima\@tempdimb
            \advance\@tempdima\@tempdimc
    \advance\@tempdimb\tabcolsep
    \advance\@tempdimc\tabcolsep
    \advance\@tempdima2\tabcolsep
            \kern-\@tempdimb
            \leaders\vrule
                    \hskip\@tempdima\@plus  1fill
            \kern-\@tempdimc
            \hskip-\wd\z@ \@plus -1fill }
\makeatother

\makeatletter
\def\hlinew#1{%
  \noalign{\ifnum0=`}\fi\hrule \@height #1 \futurelet
   \reserved@a\@xhline}
\newcolumntype{"}{@{\hskip\tabcolsep\vrule width 0.8pt\hskip\tabcolsep}}   
\makeatother

\newlength{\Oldarrayrulewidth}

\newcolumntype{L}[1]{>{\raggedright\let\newline\\\arraybackslash\hspace{0pt}}m{#1}}
\newcolumntype{C}[1]{>{\centering\let\newline\\\arraybackslash\hspace{0pt}}m{#1}}
\newcolumntype{R}[1]{>{\raggedleft\let\newline\\\arraybackslash\hspace{0pt}}m{#1}}

\definecolor{dblue}{rgb}{0,0,0.75}
\definecolor{dgreen}{rgb}{0,0.7,0}
\definecolor{lightblue}{rgb}{0.83,0.91,1.0}

\definecolor{red1}{rgb}{1, 0.3, 0.4}
\definecolor{red2}{rgb}{0.98, 0.02, 0.09}
\definecolor{red3}{rgb}{0.9, 0, 0}

\definecolor{c_l2}{rgb}{0.2, 0.42, 0.95}
\definecolor{c_l4}{rgb}{0.04, 0.7, 0.7}
\definecolor{c_l6}{rgb}{0.65, 0.65, 0.10}
\definecolor{c_lopt}{rgb}{1, 0.2, 0.2}

\definecolor{c_m2}{rgb}{0.242, 0.150, 0.660}
\definecolor{c_m4}{rgb}{0.268, 0.345, 0.954}
\definecolor{c_m6}{rgb}{0.169, 0.570, 0.936}
\definecolor{c_m8}{rgb}{0.043, 0.741, 0.739}
\definecolor{c_m10}{rgb}{0.509, 0.787, 0.350}
\definecolor{c_m12}{rgb}{0.9, 0.7, 0.2}
\definecolor{c_m14}{rgb}{1,        0.302, 0.302}

\newcommand*{\QEDB}{\hfill\ensuremath{\square}}%

\newcommand{\Ps}{\mathscr{P}}

\declaretheoremstyle[
spaceabove=4pt, spacebelow=4pt,
headfont=\normalfont\bfseries,
notefont=\normalfont, notebraces={(}{)},
headpunct={:},
bodyfont=\normalfont,
postheadspace=0.5em,
]{mynote}
\declaretheoremstyle[
spaceabove=4pt, spacebelow=4pt,
headfont=\normalfont,
notefont=\normalfont, notebraces={(}{)},
headpunct={:},
bodyfont=\normalfont,
postheadspace=0.5em,
]{mynote2}

\declaretheorem[style=mynote]{Remark}

\declaretheorem[style=mynote2,name=$\Ps$-\hspace{-1.2mm}]{Problem}
\declaretheorem[style=mynote2,name=$\Ps$-\hspace{-1.2mm}]{ProblemA}

\declaretheorem[style=mynote,name=Theorem,numberwithin=section]{ThmA}
\declaretheorem[style=mynote,name=Lemma,numberwithin=section]{LemA}

\declaretheorem[style=mynote,name=Definition,numberwithin=section]{Def}

\newcommand{\ket}[1]{\left| #1 \right>} 
\newcommand{\bra}[1]{\left< #1 \right|} 

\newcommand{\paperTitle}{Minimization of the estimation error for entanglement distribution networks\\ with arbitrary noise}

\begin{document}
\normalem

{\color{white}
\fontsize{0pt}{0pt}\selectfont
\begin{acronym}
\acro{QED}{quantum entanglement distillation}\vspace{-5.5mm}
\acro{w.r.t.}{with respect to}\vspace{-5.5mm}
\acro{i.i.d.}{independent and identically distributed}\vspace{-5.5mm}
\acro{pmf}{probability mass function}\vspace{-5.5mm}
\acro{pdf}{probability density function}\vspace{-5.5mm}
\acro{cdf}{cumulative distribution function}\vspace{-5.5mm}
\acro{iff}{if and only if}\vspace{-5.5mm}
\acro{POVM}{positive operator-valued measure}\vspace{-5.5mm}
\acro{QBER}{quantum bit error rate}\vspace{-5.5mm}
\end{acronym}}


\setcounter{page}{1}
\title{\paperTitle}

\author{Liangzhong Ruan}
\affiliation{
School of Cyber Science and Engineering, Xi'an Jiaotong University, China\\
Qutech, Delft University of Technology, The Netherlands
}



\begin{abstract}
Fidelity estimation is essential for the quality control of entanglement distribution networks. 
Because measurements collapse quantum states, 
we consider a setup in which nodes randomly sample a subset of the entangled qubit pairs to measure 
and then estimate the average fidelity of the unsampled pairs conditioned on the measurement outcome. 
The proposed estimation protocol achieves the lowest mean squared estimation error in a difficult scenario with arbitrary noise and no prior information.
Moreover, this protocol is implementation friendly because it only performs local Pauli operators according to a predefined sequence.  
Numerical studies show that compared to existing fidelity estimation protocols,
 the proposed protocol reduces the estimation error in both scenarios with i.i.d. noise and correlated noise.
\end{abstract}


\maketitle

\acresetall             

\section{Introduction}

Entanglement distribution networks \cite{WehElkHan:J18} are an important developmental stage on the way to a full-blown quantum Internet \cite{Kim:J08,Cas:J18,Panetal:J19}. 
Such networks enable device-independent protocols \cite{ArnDupFawRenVid:J18,AveMarValVil:J18,LeeHob:J18,Sch-etal:J21,Lietal:J21}, thereby achieving the highest level of quantum security \cite{Pir-etal:J20}.
Entanglement distribution networks do not require nodes to have local quantum memory.
Since efficient communication-compatible quantum memories are still under development \cite{KorMorLanNeuRitRem:J18}, entanglement distribution networks serve as a cornerstone for realizing trustworthy quantum applications with state-of-the-art quantum technology \cite{YinCaoPanetal:J17,Humetal:J18,Pom-etal:J21,LiuTianGuetal:J20}.

Entanglement quality assessment is a key building block for entanglement distribution networks.
To this end, fidelity estimation for entangled states is a promising candidate.
Fidelity is a metric that indicates the quality of quantum states \cite{WesLidFonGyu:J10,Basetal:J11,ArrLazHelBal:J14,ZhaLuoWen-etal:J21} 
and can be estimated with separable quantum measurements and classical post-processing.
Several fidelity estimation protocols have been proposed \cite{SomChiBer:J06,GuhLuGaoPan:J07,FlaLiu:J11,ZhuHay:J19}, and fidelity estimation protocols of entangled states have been implemented in several recent experiments \cite{Kaletal:J17,YinCaoPanetal:J17,Humetal:J18,Pom-etal:J21}.
However, there are still two challenges that need to be addressed.
\begin{itemize}[leftmargin=*]
\item{\em Excessive estimation error due to arbitrary noise:} 
Quantum networks often face heterogeneous and correlated noise.
When distributing quantum keys and estimating channel capacities, 
this noise leads to excessive estimation errors \cite{TomLimGisRen:J12,PfiRolManTomWeh:J18}.
Therefore, designing low-error fidelity estimation protocols in the presence of arbitrary noise is an interesting area of research.
\item{\em Efficiency loss due to separable operations:} 
In the absence of quantum memory, nodes in entanglement distribution networks perform operations that are separable between all qubits. 
Such operations result in significantly lower estimation efficiency \cite{BagBalGilRom:J06} compared to joint operations.
Minimizing the loss of efficiency due to separable operations is a major challenge for fidelity estimation.
\end{itemize}

This work focuses on fidelity estimation for entanglement distribution networks.
Since measurements collapse quantum states, we consider a network in which nodes randomly sample a subset of qubit pairs to measure
and estimate the average fidelity of unsampled pairs conditioned on the measurement outcome.

We prove that the protocol proposed in this paper achieves the lowest estimation error in the difficult scenario with arbitrary noise and no prior information, thereby overcoming the challenges listed above.
Moreover, the proposed  protocol is implementation friendly. 
This protocol uses only local Pauli measurements, standard operations that can be implemented on a variety of quantum platforms, \cite{DeGetal:J13,YinCaoPanetal:J17,Kaletal:J17,Bocetal:J18}, 
and determines the basis of each Pauli measurement according to a predefined sequence, so that no adaptive operation is required.

The remainder of the paper is organized as follows.
Section~\ref{sec:formulation} formulates the problem of minimizing the estimation error of fidelity in a scenario with arbitrary noise and no prior information. 
Section~\ref{sec:transform} and~\ref{sec:construction} solve the formulated problem and present a protocol that minimizes the estimation error of fidelity.
Section~\ref{sec:sim} evaluates the proposed protocol by comparing it with exsiting ones.
Finally, Section~\ref{sec:sum} presents the conclusion.

\section{System Setup and Problem Formulation}

\label{sec:formulation}

Consider a scenario in which two nodes share $N$ noisy qubit pairs and have no prior information of the noise.
The nodes want to estimate the fidelity of the qubit pairs \ac{w.r.t.} the target maximally entangled state.
Since maximally entangled states are mutually convertible via local operations, we set the target state to $|\Psi^-\rangle$ without loss of generality, where
$
|\Psi^\pm \rangle =\frac{|01\rangle \pm |10\rangle}{\sqrt{2}}$, 
$|\Phi^\pm \rangle =\frac{|00\rangle \pm |11\rangle}{\sqrt{2}}$
are the four Bell states. 
Denote the set of all qubit pairs by $\Set N$. The nodes randomly sample $M(< N)$ number of qubit pairs to measure.
The set of sampled pairs, $\Set{M}$, is drawn from all $M$-subsets of $\Set N$ with equal probability.

Entanglement distribution networks do not require nodes to have local quantum memory \cite{WehElkHan:J18}.
Consequently, the nodes execute operators that are separable between all qubits.  
Denote the separable operators by 
\begin{align}
\M{M}_{r}= \sum_{k}\otimes^{n\in \Set M} (\M{M}^{\mathrm{(A)}}_{r,n,k}\otimes \M{M}^{\mathrm{(B)}}_{r,n,k}),\label{eqn:MXsep}
\end{align} 
where the node index $\mathrm{X}\in\{\mathrm{A}, \mathrm{B}\}$, the qubit pair index $n\in \Set M$, and the operator index $r\in \Set R$, with 
$\Set R$ denoting the set of all possible measurement outcomes.
As \ac{POVM} operators, 
\begin{align}
\M{M}^{\mathrm{(X)}}_{r,n,k}\succcurlyeq \M{0}, \; \sum_{r}\M{M}_{r} =\mathbb{I}_{4^M}, 
\label{eqn:MPOVM}
\end{align}
where $\succcurlyeq$ denotes the matrix inequality.
The measurement operation performed on all qubit pairs can be defined as
\begin{align}
\Set{O}= \{\M{M}_{r},\; r\in \Set R \},\label{eqn:measurement}
\end{align}

Denote $\V{\rho}_{\mathrm{all}}$ as the joint state of all qubit pairs before the measurements, and
denote the $\V{\rho}_{\mathrm{all}}^{(r)}$ as the joint state of the qubit pairs conditioned on the measurement outcome $r$.
In this case, the state of the $n$-th pair is
\begin{align}
\V{\rho}^{(r)}_n = \mathrm{Tr}_{i\in{\Set N}\backslash \{n\}} \V{\rho}_{\mathrm{all}}^{(r)}.
\end{align}
Because measurements collapse quantum states, the nodes estimate the average fidelity of unsampled pairs conditioned on the measurement outcome $r$, i.e.,
\begin{align}
\bar{f}= \frac{1}{N-M}\sum_{n\in \Set N \backslash \SM }\langle\Psi^-|\V{\rho}^{(r)}_n |\Psi^-\rangle.\label{eqn:Fbar}
\end{align}
The imperfection of the measurements is considered as part of the noise.
Based on the measurement outcome, the estimator $\Set D$ is used to estimate the average fidelity of the unsampled qubit pairs $\bar{f}$, i.e.,
\begin{align}
\check{f} = {\Set D}(r).\label{eqn:estimator}
\end{align}

The nodes target at minimizing the mean squared error of the estimated fidelity.
Because the set of sampled qubit pairs $\Set M$ is drawn from all $M$-subsets of $\Set N$ with equal probability, each 
$M$-subset is selected with probability ${N\choose M}^{-1}$.
Therefore, the mean squared estimation error is given by 
\begin{align}{N\choose M}^{-1} \sum_{\Set M}\mathbb{E}_{R}\big[(\check{F}-\bar{F})^2\big].
\end{align}
Since the factor ${N\choose M}^{-1}$ does not affect the optimization of the measurement protocol $\{\Set O, \Set D\}$, it is omitted in the following problem formulation for clarity.

To ensure the robustness of the fidelity estimation protocol, we consider the scenario with arbitrary noise and no prior information.
Since there is arbitrary noise, the state of all qubit pairs $\V{\rho}_{\mathrm{all}}\in{\Set S}_{\mathrm{arb}}$, where ${\Set S}_{\mathrm{arb}}$ is the set of all $N$ qubit-pair states.
Since there is no prior information, the state $\V{\rho}_{\mathrm{all}}$ that leads to the largest error must be considered.
This optimization problem is formulated as follows.
\begin{Problem} \label{prob:1-m} 
\label{prob:1}
\end{Problem}
\vspace*{-8mm}
\begin{subequations}
\begin{align}
 \underset{\Set O, \Set D}{\mathrm{minimize}}\;\;&\max_{\V{\rho}_{\mathrm{all}}\in{\Set S}_{\mathrm{arb}}}
{N \choose M}^{\!-1\!}\sum_{\Set M}\mathbb{E}_{R}\big[(\check{F}-\bar{F})^2\big]\\
\text{subject to}\;\;
&\sum_{\SM}\mathbb{E}_{R}\big[\check{F}-\bar{F}\big]=0 ,\;\; \forall \V{\rho}_{\mathrm{all}}\in{{\Set S}_{\mathrm{arb}}}
\label{eqn:unbiased_0-m} 
\end{align}
\end{subequations}
where $\check{F}$, $\bar{F}$, and $R$ are the random variable form of the estimated fidelity $\check{f}$, the average fidelity $\bar{f}$, and the measurement outcome $r$, respectively, \eqref{eqn:unbiased_0-m} is the unbiased constraint of the estimate.

The following two sections elaborates the key procedures and ideas for solving $\Ps$-1, which
involves two main steps, namely problem transformation and operation construction.
Readers may refer to Appendices~\ref{sec:simplify} and~\ref{sec:construct} for the formal propositions and detailed derivations of these two steps.

\section{Problem Transformation}
\label{sec:transform}
This step shows that {$\Ps$-\ref{prob:1-m}} can be transformed into an equivalent problem with independent noise.
It consists of three substeps.
\begin{itemize}[leftmargin=0 mm]
\item {\bf A1.} This substep simplifies $\Ps$-\ref{prob:1-m} to an equivalent one with classical correlated noise, that is,
\begin{Problem}\label{prob:sp-m}
Special case of $\Ps$-\ref{prob:1-m}, with ${\Set S}_{\mathrm{arb}}$ replaced by ${\Set S}_{\mathrm{sp}}$, i.e., the set of $\V{\rho}_{\mathrm{all}}$ that are separable among all qubit pairs.
\end{Problem}

The key concept is to construct an operation $\Set T$ that removes the entanglement among different qubit pairs without changing the fidelity. 
This operation transforms the states in ${\Set S}_{\mathrm{arb}}$ to those in ${\Set S}_{\mathrm{sp}}$. 
Then, the transformed state can be estimated by the optimal solution of $\Ps$-\ref{prob:sp-m},
ensuring that the minimum estimation error for states in ${\Set S}_{\mathrm{arb}}$ is no higher than for states in ${\Set S}_{\mathrm{sp}}$.

Specifically,  the probabilistic rotation $\Set T$ acts independently on each qubit pair.
This operation rotates both qubits first along the $x$-axis by $180^\circ$ with probability 0.5 and
then along the $y$-axis by $180^{\circ}$ with probability 0.5.
For Bell states $|\phi\rangle, |\psi\rangle \in \{|\Phi^\pm\rangle$, $|\Psi^\pm\rangle\}$,
\begin{subequations}
  \label{eqn:R-phi-m}
    \begin{empheq}[left={\Set{T}(|\phi\rangle\langle \psi | )=\empheqlbrace\,}]{align}
      &|\phi\rangle\langle \psi | &\hspace{-13mm}\mbox{if } |\phi\rangle = |\psi\rangle
        \label{eqn:R-phi-m-1} \\
      &\M{0} &\hspace{-13mm} \mbox{if } |\phi\rangle \neq |\psi\rangle
        \label{eqn:R-phi-m-2}
    \end{empheq}
\end{subequations}
\eqref{eqn:R-phi-m-1} ensures that $\Set{T}$ does not change the fidelity of the qubit pairs.
\eqref{eqn:R-phi-m-2} shows that $\Set{T}$ removes the off-diagonal terms of a density matrix expressed in Bell basis.
Such a diagonal density matrix corresponds to states separable among all qubit pairs.

\item {\bf A2.} This substep shows that to solve $\Ps$-\ref{prob:sp-m}, it is sufficient to consider the special case with independent noise, i.e.,
\begin{Problem}\label{prob:id-m}
Special case of $\Ps$-\ref{prob:sp-m}, with ${\Set S}_{\mathrm{sp}}$ replaced by ${\Set S}_{\mathrm{id}}$, i.e., the set of $\V{\rho}_{\mathrm{all}}$ that are product states.
\end{Problem}

Denote $\check{f}(r)$ as the estimated fidelity given the measurement outcome $r$, 
and denote $\bar{f}(\V{\rho}_{\mathrm{all}},\Set{M}, r)$ as the average fidelity given the state $\V{\rho}_{\mathrm{all}}$, the sample set $\SM$ and the measurement outcome $r$.
Separable states $\V{\rho}_{\mathrm{all}}\in {\Set S}_{\mathrm{sp}}$ are ensembles of product states, i.e., 
$
\V{\rho}_{\mathrm{all}} = \sum_{k}p_k\V{\rho}^{(k)}_{\mathrm{all}}$,
where $\V{\rho}^{(k)}_{\mathrm{all}}\in {\Set S}_{\mathrm{id}}$, and ensemble probabilities $p_k$ satisfy $p_k\geq 0$, $\sum_{k}p_k=1$.
Consequently, for each set $\SM$, the estimation error with a separable state $\V{\rho}_{\mathrm{all}}$ and that with product states $\{\V{\rho}^{(k)}_{\mathrm{all}}\}$ are related by the following inequality.
\begin{subequations} 
\begin{align}
&\mathbb{E}_{R}\big[(\check{F}-\bar{F})^2\big |\V{\rho}_{\mathrm{all}} \big]\nonumber \\
&=  \sum_{r}\mathrm{Pr}(r|\SM)
\big(\check{f}(r) -\bar{f}(\V{\rho}_{\mathrm{all}},\Set{M}, r)\big)^2\label{eqn:sp2id-a-m}\\
&= \sum_{r}\mathrm{Pr}(r|\SM)\big(\check{f}(r) - \sum_{k }
\mathrm{Pr}(\V{\rho}^{(k)}_{\mathrm{all}}| \SM, r) \bar{f}(\V{\rho}^{(k)}_{\mathrm{all}},\Set{M}, r) \big)^2 \label{eqn:sp2id-b-m}\\
&\leq \sum_{r,k}
\mathrm{Pr}(r|\SM)\mathrm{Pr}(\V{\rho}^{(k)}_{\mathrm{all}}| \SM, r)
\big(\check{f}(r) -\bar{f}(\V{\rho}^{(k)}_{\mathrm{all}},\Set{M}, r) \big)^2 \label{eqn:sp2id-c-m}\\
&= \sum_{k}p_k\!\sum_{r}
\mathrm{Pr}( r|\V{\rho}^{(k)}_{\mathrm{all}},\SM)
\big(\check{f}(r) \hspace{-0.2mm}-\hspace{-0.4mm} \bar{f}(\V{\rho}^{(k)}_{\mathrm{all}},\Set{M}, r)\big)^{\hspace{-0.2mm}2\hspace{-0.5mm}}\label{eqn:sp2id-d-m}\\
&= \sum_{k}p_k\mathbb{E}_{R}\big[(\check{F}-\bar{f})^2\big|\V{\rho}^{(k)}_{\mathrm{all}}\big]. \label{eqn:sp2id-e-m}
\end{align}\label{eqn:sp2id-m}
\end{subequations}
where $\mathrm{Pr}(\cdot)$ denotes the probability, \eqref{eqn:sp2id-c-m} is true because
$(\sum_{k}w_kx_k)^2 \leq \sum_{k}w_kx^2_k$,
$\forall\, x_k \in \mathbb{R}, w_k\geq 0, \sum_{k}w_k=1$, and \eqref{eqn:sp2id-d-m} holds because according to Bayes' theorem
\begin{align}
\mathrm{Pr}(r|\SM)\mathrm{Pr}(\V{\rho}^{(k)}_{\mathrm{all}}| \SM, r) &=\mathrm{Pr}(\V{\rho}^{(k)}_{\mathrm{all}}|\SM)\mathrm{Pr}(r|\V{\rho}^{(k)}_{\mathrm{all}},\SM)
\nonumber\\
&=\mathrm{Pr}(\V{\rho}^{(k)}_{\mathrm{all}})\mathrm{Pr}(r|\V{\rho}^{(k)}_{\mathrm{all}},\SM)\nonumber\\
&=p_k \mathrm{Pr}(r|\V{\rho}^{(k)}_{\mathrm{all}},\SM),\label{eqn:Bayes}
\end{align}
where the second equation holds becase that the sampling set $\Set M$ and the state $\V{\rho}^{(k)}_{\mathrm{all}}$ are independent.
Because $\V{\rho}_{\mathrm{all}}\in \Set S_{\mathrm{sp}}$ and $\V{\rho}^{(k)}_{\mathrm{all}}\in \Set S_{\mathrm{id}}$, $\forall k$,
\eqref{eqn:sp2id-m} shows that the minimum estimation error of $\Ps$-\ref{prob:sp-m} is upper bounded by that of $\Ps$-\ref{prob:id-m}.

\begin{Remark}[Limit the effect of correlation]
\label{remark:neutralize}
The result of step A2 is unexpected because correlated noise usually leads to a higher estimation error compared to independent noise.
To understand this result, recall that the estimation target, i.e., the average fidelity of the unsampled qubit pairs, is evaluated conditioned on the measurement outcome.
Consequently,  the conditional distribution $\mathrm{Pr}(\V{\rho}^{(k)}_{\mathrm{all}}| \SM, r)$ appears in the expression of the estimation error, i.e., \eqref{eqn:sp2id-b-m}.
This conditional probability allows the application of Bayes' theorem in \eqref{eqn:Bayes}, 
which bounds the estimation error of the states in $\Set S_{\mathrm{sp}}$ by that of the states in $\Set S_{\mathrm{id}}$.~\hfill\qed
\end{Remark}

\item{\bf A3.} This substep shows that to solve $\Ps$-\ref{prob:id-m},
it is sufficient to minimize the estimation error of the sampled pairs, i.e.,

\begin{Problem}\label{prob:id2-m}
\end{Problem}
\vspace*{-8mm}
\begin{align}
 \underset{\Set O, \Set D}{\mathrm{minimize}}\;\;& \max_{\scriptstyle \V{\rho}_{\mathrm{all}}\in\atop \scriptstyle{\Set S}_{\mathrm{id}}(\V{f}_{\mathrm{all}})}
 {N\choose M}^{\!-1\!}\sum_{\SM}\mathbb{E}_{R}\big[(\check{F}-\bar{f}_{\SM})^2\big] \nonumber \\
\text{subject to}\;\;
&\mathbb{E}_{R}\big[\check{F}-\bar{f}_{\SM} \big]=0,\, \forall \V{\rho}_{\mathrm{all}}\in{\Set S}_{\mathrm{id}}(\V{f}_{\mathrm{all}}),  \SM \subset \Set N\nonumber
\end{align}
where ${\Set S}_{\mathrm{id}}(\V{f}_{\mathrm{all}})=\{\otimes^{n\in \Set N }\V{\rho}_n\}$ is the set of product states with fidelity composition $\V{f}_{\mathrm{all}}=\{f_n,n\in \Set N\}$, in which $f_n $ is the fidelity of $\V{\rho}_n$, and
$ \bar{f}_{\SM}$ is the average fidelity of the sampled qubit pairs.

In the case of independent noise, the measurement does not affect the fidelity of the unsampled qubit pairs.
Consequently, the estimation error of $\Ps$-\ref{prob:id-m} can be decomposed into two parts, namely the estimation error of the sampled qubit pairs,
\begin{align}
{N\choose M}^{\!-1\!}\sum_{\SM}\mathbb{E}_{R}\big[(\check{F}-\bar{f}_{\SM} )^2\big],\label{eqn:error-measure}
\end{align}
and the sampling error, i.e., the deviation between the average fidelity of the sampled and the unsampled qubit pairs,
\begin{align}
{N\choose M}^{\!-1\!}\sum_{\SM}(\bar{f}_{\SM} -\bar{f})^2.\,\label{eqn:error-sample}
\end{align}
The sampling error \eqref{eqn:error-sample} is not affected by the estimation protocol.
Therefore, $\Ps$-\ref{prob:id-m} can be simplified to $\Ps$-\ref{prob:id2-m}, which minimizes \eqref{eqn:error-measure}.
\end{itemize}

\section{Operation Construction}
\label{sec:construction}

The next step is to construct the optimal solution of $\Ps$-\ref{prob:1-m}, which involves the following two substeps.
\begin{itemize}[leftmargin=0mm]
\item{\bf B1.} This substep characterizes the minimum estimation error of $\Ps$-\ref{prob:id2-m}.

The characterization of the minimal estimation error has been a subject of intense research.
For $\Ps$-\ref{prob:id2-m}, the Cram\'er-Rao bound \cite{Cra:B99,Rao:B94} characterizes the lower bound of the estimation error for a given measurement operation $\Set O$,
and the quantum Fisher information \cite{Paris:09,Saf:18} identifies this lower bound under the condition that all measurement operations are available.

However, $\Ps$-\ref{prob:id2-m} aims to minimize the estimation error when all separable measurement operations are available.
Therefore, the Cram\'er-Rao bound alone does not provide the lowest bound, while the quantum Fisher information provides a lower bound that is infeasible.
We will close this gap by first further simplifying $\Ps$-\ref{prob:id2}, 
and then characterizing the limit of separable operators.

The derivation of a Cram\'er-Rao bound \cite{Cra:B99,Rao:B94} requires knowledge of the distribution of a measurement outcome.
In $\Ps$-\ref{prob:id2-m}, this distribution is determined by the state of the qubit pairs $\V{\rho}_{\mathrm{all}}$ and the measurement operation $\Set O$.
However, even with independent noise, the state of each qubit pair $\V{\rho}_{n}\in \Set H^{4\times 4}$ has several parameters other than the fidelity.
This fact complicates the expression of the measurement outcome and makes the analysis of the Cram\'er-Rao bound not feasible.
To overcome this challenge, we show that the minimum estimation error of general independent states is bounded below by that of independent Werner states, i.e., 
\begin{Problem} \label{prob:idw-m}
Special case of $\Ps$-\ref{prob:id2-m}, with ${\Set S}_{\mathrm{id}}(\V{f}_{\mathrm{all}})$ replaced by 
${\Set S}_{\mathrm{w}}(\V{f}_{\mathrm{all}}) = \{\otimes^{n\in \Set N }\V{\sigma}_n\}$,
where  
\begin{align}
\V{\sigma}_n = f_n|\Psi^-\rangle\langle\Psi^-| + \frac{1-f_n}{3}(\mathbb{I}_4-|\Psi^-\rangle\langle\Psi^-|).\label{eqn:sigma-m}
\end{align}
\end{Problem}

Specifically, recall the bilateral rotation operation $\Set B$ in \cite{BenBraPopSchSmoWoo:J96}, which transforms generic states of a qubit pair into the Werner state with the same fidelity.
Using logic similar to that in step~A1, one can see that because the states in ${\Set S}_{\mathrm{id}}(\V{f}_{\mathrm{all}})$
can be transformed into those in ${\Set S}_{\mathrm{w}}(\V{f}_{\mathrm{all}})$ via the separable operation $\Set B$, the minimum estimation error of $\Ps$-\ref{prob:id2-m} is no higher than that of $\Ps$-\ref{prob:idw-m}.
The states in ${\Set S}_{\mathrm{w}}(\V{f}_{\mathrm{all}})$ are fully parameterized by the fidelity composition $\V{f}_{\mathrm{all}}$, which allows the Cram\'er-Rao bound analysis of $\Ps$-\ref{prob:idw-m}.

The other factor affecting the distribution of measurement outcome is the separable measurement operation $\Set O$.
Since the fidelity of each separable state \ac{w.r.t.} a maximally entangled state lies in the interval $[0,\frac{1}{2}]$ \cite{HorHorHorHor:J09}, we have that
for any separable operator $\M{M}$,
\begin{align}
0\leq \frac{\mathrm{Tr}(|\Psi^-\rangle\langle\Psi^-|\M{M})}{\mathrm{Tr}(\M{M} )}\leq \frac{1}{2}.\label{eqn:maxcor-m}
\end{align} 
For all separable operators, \eqref{eqn:maxcor-m} limits the sensitivity of the measurement outcome to the changes in fidelity.
Plugging this result into the Cram\'er-Rao bound, we lower bound the estimation error of $\Ps$-\ref{prob:idw-m} by
\begin{align}
\sum_{n\in\Set N}\frac{(2f_n+1)(1-f_n)}{2MN}.
\label{eqn:errorbound-m}
\end{align}

\item{\bf B2.} The second substep constructs an estimation protocol that is the optimal solution of $\Ps$-\ref{prob:1-m}.

\begin{figure}
\vspace{-2mm}
\begin{algorithm}[H]
\caption{Fidelity estimation}\label{alg:fidelityest}
\begin{algorithmic}[1]
\State{\em Preprocessing.} The nodes select the sample set $\SM$ completely at random, i.e,  select $\Set{M}$ from all $M$-subsets of $\Set N$ with equal probability, and generate $M$ number of \ac{i.i.d.} random variables $A_n \in\{x,y,z\}$, $n\in \SM$,
with distribution 
$
\Pr[A_n=u]=\frac{1}{3}, \quad u\in\{x,y,z\}.
$
\State{\em Perform measurements.}
For the qubit pair $n\in \Set{M}$,  both nodes measure the qubit in the $A_n$-basis.
If the measurement results of the two nodes match, record measurement outcome  $r_{n} = 1$, otherwise record $r_{n} = 0$.
\State{\em Estimate fidelity.} The number of errors and the \ac{QBER} are expressed as
$
e_{\SM} = \sum_{n\in\SM} r_{n}$, and
$\varepsilon_{\SM}= \dfrac{e_{\SM}}{M}$,
respectively. The estimated fidelity is
\begin{align}
\check{f}= 1- \frac{3}{2}\varepsilon_{\SM}.\label{eqn:checkF}
\end{align}
\end{algorithmic}
\end{algorithm}
\end{figure}

The analysis in step~B1 shows that to achieve the Cram\'er-Rao bound, each measurement operator must balance either of the two inequalities in \eqref{eqn:maxcor-m}
 (For details, see the proof of \thref{thm:min}, the derivation from \eqref{eqn:Tr-rnk} to \eqref{eqn:error-combine-2}).
According to this principle, the measurement operation $\Set O^*$ of Protocol~\ref{alg:fidelityest} is constructed.
For example, when nodes measure in the $z$-basis, the operators corresponding to $r=0$ and $r=1$ are respectively
\begin{align}
\begin{split}
&\hspace{-2mm}\M{M}_{z,0}\hspace{-0.2mm}=\hspace{-0.2mm}\ket{01}\!\bra{01} + \ket{10}\!\bra{10}\hspace{-0.2mm}=\hspace{-0.2mm} |\Psi^-\rangle\hspace{-0.3mm}\langle\Psi^-| + |\Psi^+\rangle\hspace{-0.3mm}\langle\Psi^+|,\\
&\hspace{-2mm}\M{M}_{z,1}\hspace{-0.2mm}=\hspace{-0.2mm}\ket{00}\!\bra{00} + \ket{11}\!\bra{11}\hspace{-0.2mm}=\hspace{-0.2mm} |\Phi^-\rangle\langle\Phi^-| + |\Phi^+\rangle\langle\Phi^+|.
\end{split} \label{eqn:protocloperators-m}
\end{align}
According to \eqref{eqn:protocloperators-m}, $\M{M}_{z,0}$ and $\M{M}_{z,1}$ balance the second and first inequalities of \eqref{eqn:maxcor-m}, respectively.
Moreover, according to the preprocessing step, the basis of measurement on each qubit pair are chosen in an \ac{i.i.d.} manner.
With these properties, we show that in scenarios with independent noise, the measurement outcome $R_n\in \Set M$ are independent random variables with variance 
\begin{align}
\frac{2(2f_n+1)(1-f_n)}{9}.\label{eqn:VRn-m}
\end{align}
Further noticing that the sample set $\Set{M}$ is selected from all $M$-subsets of $\Set N$ with equal probability, we substitute \eqref{eqn:VRn-m} into \eqref{eqn:checkF} and obtain that the operation $\Set O^*$ achieves the minimum estimation error given by \eqref{eqn:errorbound-m} for states in $\Set{S}_{\mathrm{id}}$.
Therefore,  $\Set O^*$ is optimal for $\Ps$-\ref{prob:id2-m}.
In this case, according to the problem transformation step, the composite operation $\hat{\Set O}^* = \Set O^* \circ \Set T$ is optimal for $\Ps$-\ref{prob:1-m}.

To further simplify the measurement operation, we note that according to \eqref{eqn:R-phi-m}, 
the operation $\Set{T}$ on each qubit pair can be expressed by the following four Kraus operators
\begin{align}
\M{T}_{\phi} &=  |\phi\rangle\langle \phi |, \quad \mbox{where}\quad \phi \in \{\Phi^\pm, \Psi^\pm\}.\label{eqn:KrausT-m}
\end{align}
Substituting \eqref{eqn:KrausT-m} into \eqref{eqn:protocloperators-m} shows that the operation $\Set T$ does not change the operators of $\Set O^*$, e.g.,
\begin{align}
\M{M}_{z,r} &= \sum_{\phi \in \{\Phi^\pm, \Psi^\pm\}} \M{T}^\dag_{\phi}\M{M}_{z,r}\M{T}_{\phi},\;\; r\in\{0,1\},\label{eqn:TMur-m}
\end{align}
\eqref{eqn:TMur-m} shows that  
$
\hat{\Set O}^* = \Set O^* \circ \Set T =\Set O^*
$,
i.e., $\hat{\Set O}^*$ and $\Set O^*$ are equivalent.
Therefore, $\Set O^*$ is optimal for $\Ps$-\ref{prob:1-m}.
\end{itemize}

\begin{Remark}[The advantage of Protocol~\ref{alg:fidelityest}]
It may seem that using bilateral Pauli measurements as in Protocol~\ref{alg:fidelityest} is a quite natural choice for estimating the fidelity of entangled qubit pairs.
In fact, using these standard measurements to obtain optimal performance in scenarios with arbitrary noise is a key advantage of Protocol~\ref{alg:fidelityest}.

To ensure implementability of Protocol~\ref{alg:fidelityest}, simpler operations are chosen with priority.
In fact, by using more complicated operations, one can obtain optimal solutions of $\Ps$-\ref{prob:1} other than Protocol~\ref{alg:fidelityest}.
For example,  consider a protocol in which the nodes perform the operation $\Set{T}$ defined in \eqref{eqn:R-phi-m} before making the Pauli measurements in Protocol~\ref{alg:fidelityest}. 
Then according to the paragraph after \eqref{eqn:protocloperators-m}, this protocol is also optimal for $\Ps$-\ref{prob:1}.
We have made additional efforts to simplify the measurement operators while preserving the optimality of the protocol.

To achieve optimal performance in scenarios with arbitrary noise, the preprocessing step, i.e., step 1, is developed.
The complete randomness of the sampling set $\Set M$, and the \ac{i.i.d.} distribution of the bases of the measurements are necessary for handling arbitrary noise.
For example, when the measurements in the $x$-, $y$-, or $z$-basis are made in a clustered manner, the estimated fidelity is biased or has a higher estimation error for some non-\ac{i.i.d.} noises.
The preprocessing in step 1 neutralizes the effect of the arbitrary nose without compromising the implementability of Protocol~\ref{alg:fidelityest}.
~\qed\end{Remark}

\section{Protocol Evaluation}
\label{sec:sim}
In this section, we evaluate the proposed protocol for fidelity estimation by comparing it to existing fidelity estimation protocols.

\noindent {\bf Noise model:}
\label{subsec:perf_eva}
The noise is modeled as correlated and heterogeneous depolarizing channels.
Specifically, the corresponding density matrix of all qubit pairs is
\begin{align}
\V{\rho}_{\mathrm{all}} = \frac{1}{2}\Big(&
\big(\otimes^{\frac{N}{4}} \V{\rho}^{(g)}\big) \otimes \big(\otimes^{\frac{3N}{4}} \V{\rho}^{(b)}\big)+\nonumber\\
&\big(\otimes^{\frac{3N}{4}} \V{\rho}^{(g)}\big) \otimes \big(\otimes^{\frac{N}{4}} \V{\rho}^{(b)}\big)\Big), 
\end{align}
where the state of each qubit pair
\begin{align}
\V{\rho}^{(s)} =  p^{(s)}\frac{\mathbb{I}_4}{4} + (1-p^{(s)})|\Psi^-\rangle\langle\Psi^-|, \quad s\in\{g,b\}
\end{align}
in which the error probabilities $p^{(s)}$, $s\in\{g,b\}$ of good and bad channels satisfy $0\leq p^{(g)} \leq  p^{(b)} \leq 1$,
and $N$ is a multiple of $4$.

The mean of the error probabilities of the bad and good channels, i.e., 
\begin{align}
p = \frac{p^{(b)}+p^{(g)}}{2}\in[0,1],
\end{align}
represents the noisy intensity, and the difference between the two probabilities, i.e.,
\begin{align}
d= p^{(b)}-p^{(g)}\in[0,1],
\end{align} 
represents the degree of correlation and heterogeneity of the noise.
In particular,  the noise is \ac{i.i.d.}~when $d=0$. 
~\QEDB

\begin{figure} \centering
\psfrag{0}[Br][Br][0.65]{0\hspace{0.3mm}}
\psfrag{0.2}[Br][Br][0.65]{0.2\hspace{0.3mm}}
\psfrag{0.4}[Br][Br][0.65]{0.4\hspace{0.3mm}}
\psfrag{0.6}[Br][Br][0.65]{0.6\hspace{0.3mm}}
\psfrag{0.8}[Br][Br][0.65]{0.8\hspace{0.3mm}}
\psfrag{0.9}[Br][Br][0.65]{0.9\hspace{0.3mm}}
\psfrag{1}[Br][Br][0.65]{1.0\hspace{-0.2mm}}
\psfrag{1.1}[Br][Br][0.65]{1.1\hspace{-0.2mm}}
\psfrag{1.2}[Br][Br][0.65]{1.2\hspace{0.3mm}}
\psfrag{1.3}[Br][Br][0.65]{1.3\hspace{0.3mm}}
\psfrag{1.4}[Br][Br][0.65]{1.4\hspace{0.3mm}}
\psfrag{1.5}[Br][Br][0.65]{1.5\hspace{0.3mm}}
\psfrag{1.6}[Br][Br][0.65]{1.6\hspace{0.3mm}}
\psfrag{1.8}[Br][Br][0.65]{1.8\hspace{0.3mm}}
\psfrag{2}[Br][Br][0.65]{2.0\hspace{0.3mm}}
\psfrag{10-3}[bc][bc][0.7]{$\times 10^{-3}$}
\psfrag{1.0}[tt][tt][0.7]{1}
\psfrag{0.90}[tt][tt][0.7]{0.9}
\psfrag{0.80}[tt][tt][0.7]{0.8}
\psfrag{0.70}[tt][tt][0.7]{0.7}
\psfrag{0.60}[tt][tt][0.7]{0.6}
\psfrag{0.50}[tt][tt][0.7]{0.5}
\psfrag{0.40}[tt][tt][0.7]{0.4}
\psfrag{0.30}[tt][tt][0.7]{0.3}
\psfrag{0.20}[tt][tt][0.7]{0.2}
\psfrag{0}[tt][tt][0.7]{0}
\psfrag{P}[tc][tc][0.8]{$p$}
\psfrag{Dif}[tc][tc][0.8]{$d$}
\psfrag{A}[tc][tc][0.8]{A}
\psfrag{B}[tc][tc][0.8]{B}
\psfrag{Var}[bc][cc][0.8]{$\mathrm{Var}(\check{F})$}
\psfrag{proposed                                                  1}[bl][bl][0.73]{\hspace{-0mm}proposed protocol}
\psfrag{Guh}[cl][cl][0.73]{\hspace{-0mm}protocol proposed in  \cite{GuhLuGaoPan:J07}}
\psfrag{Fla}[cl][cl][0.73]{\hspace{-0mm}protocol proposed in  \cite{FlaLiu:J11}}
\hspace*{-4mm}\includegraphics[scale=0.48]{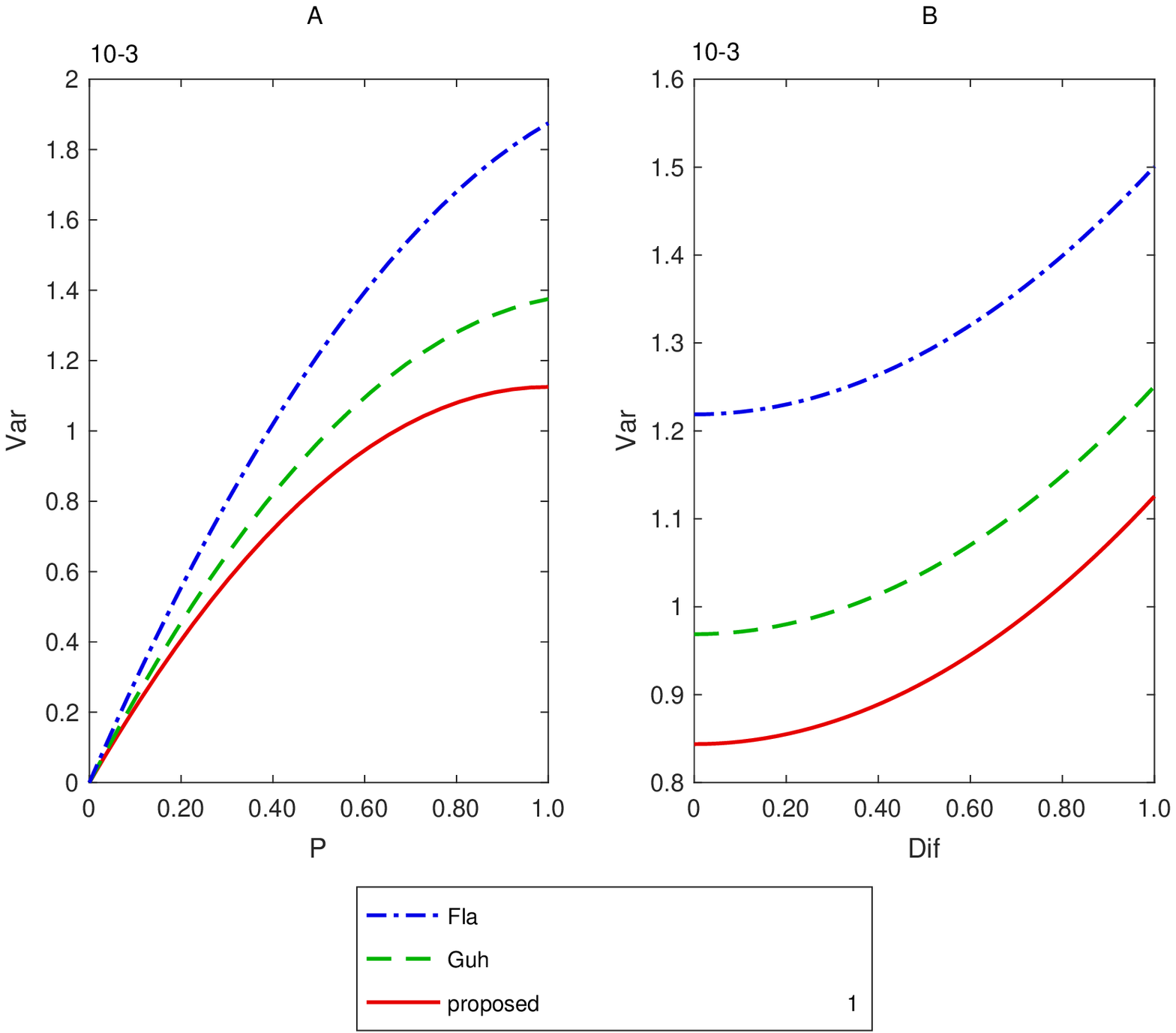}
\caption {\small The variance of the estimated fidelity given by different protocols as a function of the noisy intensity $p$ and the degree of correlation and heterogeneity $d$. 
In this figure, $N=1000$, $M=500$. In subfigure A, $d=0$, i.e., the noise is \ac{i.i.d.}. In subfigure B, $p=0.5$.
}
\label{fig_Measure_Var}
\end{figure}

Fig.~\ref{fig_Measure_Var} shows the mean squared estimation error, i.e., $\mathrm{Var}(\check{F})$, as a function of the noisy intensity $p$ and the degree of correlation and heterogeneity $d$.
The proposed protocol is compared with those proposed in  \cite{FlaLiu:J11} and \cite{GuhLuGaoPan:J07}.
This figure shows that the estimation error is an increasing function of both $p$ and $d$.
In all cases, the proposed protocol has the lowest estimation error. 
This result is consistent with \thref{thm:opt}, which states that Protocol~\ref{alg:fidelityest} is an optimal solution of $\Ps$-\ref{prob:1-m}.

\section{Conclusion}
\label{sec:sum}
In this paper, we propose a protocol to estimate the fidelity of entangled qubit pairs shared by remote nodes.
The proposed protocol uses only Pauli measurements, and 
achieves the minimum mean squared error in a challenging scenario with arbitrary noise and no prior information.
The numerical tests confirm the efficiency and reliability of the proposed protocol.

\section {Acknowledgement}
 
The author would like to thank Stephanie Wehner and Bas Dirkse at TUDelft, The Netherlands, for refining the structure of the paper and stimulating discussions on the treatment of non-\ac{i.i.d.} noise.
The author is also indebted to Wenhan Dai at UMass, USA, who distilled the logic of the derivation and proofread the manuscript.

The author thanks the EU Flagship on Quantum Technologies and the Quantum Internet Alliance (funded by the EU Horizon 2020 Research and Innovation Programme, Grant No. 820445).

\appendix
\renewcommand\thesubsection{\thesection\arabic{subsection}}

\section*{Notations}
In this document, random variables and their realizations are represented in upper and lower case letters, respectively, e.g. $F$ and $f$. 
Vectors and matrices are denoted by bold letters, e.g., $\V{\rho}$. 
$\mathbb{E}[\cdot]$, and $\mathbb{V}[\cdot]$ denote the expectation and variance of a random variable, respectively.
$\mathbb{I}_n$ denotes an $n\times n$ identity matrix.
$\mathrm{Pr}(\cdot)$ denotes the probability of an event,
$\mathrm{Tr}(\cdot)$ denotes the trace of a matrix, and $\mathrm{Tr}_{n}(\cdot)$ denotes the partial trace of the $n$-th subsystem.
\label{sec:symbol}

\section{Problem Transformation}
\label{sec:simplify}
Recall the problem of minimizing the mean squared estimation error of fidelity in scenarios with arbitrary noise, i.e.,
\begin{ProblemA}[Error minimization for arbitrary states] 
\label{prob:1}
\begin{subequations}
\begin{align}
 \underset{\Set O, \Set D}{\mathrm{minimize}}\;\;&\max_{\V{\rho}_{\mathrm{all}}\in{\Set S}_{\mathrm{arb}}}
{N \choose M}^{\!-1\!}\sum_{\Set M}\mathbb{E}_{R}\big[(\check{F}-\bar{F})^2\big] \label{eqn:minvar_0}\\
\text{subject to}\;\;
&\sum_{\SM}\mathbb{E}_{R}\big[\check{F}-\bar{F}\big]=0 ,\;\; \forall \V{\rho}_{\mathrm{all}}\in{{\Set S}_{\mathrm{arb}}}
\label{eqn:unbiased_0} 
\end{align}
\end{subequations}
\end{ProblemA}
\noindent In this section, we will prove that in order to solve $\Ps$-\ref{prob:1}, it suffices to consider the scenario with independent noise.

\subsection{Sufficiency of considering classically correlated noise}
We first show that to solve $\Ps$-\ref{prob:1}, it is sufficient to consider a problem with no entanglement among different qubit pairs.
This problem is given by $\Ps$-\ref{prob:sp-m}. 
\begin{ProblemA}[Error minimization for separable states] \label{prob:sp}
\begin{subequations}
\begin{align}
 \underset{\Set O, \Set D}{\mathrm{minimize}}\;\;&\max_{\V{\rho}_{\mathrm{all}}\in{\Set S}_{\mathrm{sp}}}\sum_{\SM}\mathbb{E}_{R}\big[(\check{F}-\bar{F})^2\big] 
\label{eqn:minvar_sep}\\
\text{subject to}\;\;
&\sum_{\SM}\mathbb{E}_{R}\big[\check{F}-\bar{F}\big]=0 ,\;\; \forall \V{\rho}_{\mathrm{all}}\in{\Set S}_{\mathrm{sp}}\label{eqn:unbiased_sep}
\end{align}
\end{subequations}
where
${\Set S}_{\mathrm{sp}}$ denotes the set of states $\V{\rho}_{\mathrm{all}}$ that are separable among all qubit pairs.
\end{ProblemA}

We construct a probabilistic rotation operation as follows.
\begin{Def}[Probabilistic bilateral rotation]\thlabel{def:T}
The operation $\Set T$ acts independently on each qubit pair.
For each pair, $\Set T$ rotates both qubits first along the $x$-axis for $180^\circ$ with probability 0.5,
then along the $y$-axis for $180^\circ$ with probability 0.5.
The Kraus operators of the first and second probabilistic rotations are
\begin{subequations}
\begin{align}
\M{T}_{x,1} &= \frac{\mathbb{I}_4}{\sqrt{2}}, \quad \M{T}_{x,2} = -\frac{\V{\sigma}_x \otimes \V{\sigma}_x}{\sqrt{2}},\quad \mbox{ and}
\label{eqn:Rx}
\\
\M{T}_{y,1} &= \frac{\mathbb{I}_4}{\sqrt{2}}, \quad \M{T}_{y,2} = -\frac{\V{\sigma}_y \otimes \V{\sigma}_y}{\sqrt{2}},
\label{eqn:Ry}
\end{align}\label{eqn:Rxy}
\end{subequations}
respectively.
\end{Def}

The following lemma shows the equivalence of  $\Ps$-\ref{prob:1} and $\Ps$-\ref{prob:sp}.

\begin{LemA}[Equivalence of $\Ps$-\ref{prob:1} and $\Ps$-\ref{prob:sp}] \thlabel{lem:sp2arb}
If a measurement operation $\Set O^*$ is optimal in $\Ps$-\ref{prob:sp}, 
the composite operation $\hat{\Set O}^* = \Set O^* \circ \Set T$ is optimal in $\Ps$-\ref{prob:1}.
\end{LemA}
\begin{proof}
The four Bell states $\{|\Phi^\pm\rangle, |\Psi^\pm\rangle\}$ form a basis for a qubit pair. 
Therefore, the state of each qubit pair $n \in \Set N$ can be written as
\begin{align}
\V{\rho}_n = \sum_{k\in \Set K_n} c_{k,n} |\phi_{k}\rangle\langle \psi_{k} |,
\label{eqn:Rhon}
\end{align}
where $|\phi_{k}\rangle, |\psi_{k}\rangle$ are Bell states, i.e., $|\Phi^\pm\rangle$, $|\Psi^\pm\rangle$, and $\Set K_n$ is the set of all terms with non-zero coefficient $c_{k,n}\in\mathbb{C}$ .
Given \eqref{eqn:Rhon}, the fidelity of qubit pair $n$ is given by
\begin{align}
f_{n} = \sum_{k\in \Set K_n}&\Big( c_{k,n}
\mathds{1}\big(|\phi_{k}\rangle= |\psi_{k}\rangle= |\Psi^-\rangle \big) \Big) 
\label{eqn:fid-n}
\end{align}
where $\mathds{1}(\cdot)$ is the indicator function.

For a Bell state $|\phi\rangle \in \{|\Phi^\pm\rangle$, $|\Psi^\pm\rangle\}$,
\begin{subequations}
\begin{align}
\V{\sigma}_x \otimes \V{\sigma}_x |\phi\rangle &= 
\left\{\begin{array}{rl}
|\phi\rangle & \mbox{ if } |\phi\rangle \in \{ |\Phi^+\rangle$, $|\Psi^+\rangle \}\\
-|\phi\rangle & \mbox{ if } |\phi\rangle \in \{ |\Phi^-\rangle$, $|\Psi^-\rangle \}
\end{array}\right.\\
\V{\sigma}_y \otimes \V{\sigma}_y |\phi\rangle &= 
\left\{\begin{array}{rl}
|\phi\rangle & \mbox{ if } |\phi\rangle \in \{ |\Phi^-\rangle$, $|\Psi^+\rangle \}\\
-|\phi\rangle & \mbox{ if } |\phi\rangle \in \{ |\Phi^+\rangle$, $|\Psi^-\rangle \}
\end{array}\right.
\end{align}\label{eqn:sigma-phi}
\end{subequations}
According to \eqref{eqn:Rxy} and \eqref{eqn:sigma-phi}, for Bell states $|\phi\rangle, |\psi\rangle \in \{|\Phi^\pm\rangle$, $|\Psi^\pm\rangle\}$,
\begin{align}
\Set{T}(|\phi\rangle\langle \psi | ) &= 
\left\{\begin{array}{ll}
|\phi\rangle\langle \psi | & \mbox{ if } |\phi\rangle = |\psi\rangle \\
0 & \mbox{ if } |\phi\rangle \neq |\psi\rangle
\end{array}\right.\label{eqn:R-phi}
\end{align}
By substituting \eqref{eqn:R-phi} into \eqref{eqn:fid-n}, it is clear that the rotation $\Set T$ does not change the fidelity of any qubit pair.                                                                                                                                            

Consider a composite operation in which the rotation $\Set T$ is performed on the sampled qubit pairs $n\in \SM$ before the measurement $\Set O$.
Subsequently, this rotation is performed on the unsampled qubit pairs $n \in \Set N \backslash \SM$. 
According to the analysis described in the previous paragraph, the rotation $\Set T$ on the unsampled qubit pairs does not affect the measurement outcome $\V{r}$ nor the average fidelity of the unsampled qubit pairs $\bar{f}$.
In this sense, the considered composite operation is equivalent to the operation $\hat{\Set O} = \Set O \circ \Set T$, i.e., both the rotation and the measurement are performed only on the sampled qubit pairs.
Moreover, since the rotation $\Set T$ on the unsampled qubit pairs is commutable with the measurement  $\Set O$, the operation $\hat{\Set O}$ is also equivalent to an operation in which the rotation $\Set T$ is performed on all qubit pairs before the measurement $\Set O$.
Given the above two equivalences, when analyzing the performance of operation $\hat{\Set O}$, we consider the operation in which the rotation is performed on all qubit pairs before the measurement.

An arbitrary state of all $N$ qubit pairs can be expressed in Bell basis as follows.
\begin{align}
\V{\rho}_{\mathrm{all}} = \sum_{k\in \Set K} c_k\otimes^{n\in\Set N} |\phi_{k,n}\rangle\langle \psi_{k,n} |,
\label{eqn:Rhoall}
\end{align}
where $|\phi_{k,n}\rangle, |\psi_{k,n}\rangle$ are Bell states, i.e., $|\Phi^\pm\rangle$, $|\Psi^\pm\rangle$, and $\Set K$ is the set of all terms with non-zero coefficient $c_k\in\mathbb{C}$.
Substituting \eqref{eqn:R-phi} into \eqref{eqn:Rhoall}, we find that after performing the rotation $\Set T$ on all qubit pairs, their joint state is
\begin{align}
\Set T(\V{\rho}_{\mathrm{all}}) = \sum_{k\in \Set K} c_k\otimes^{n\in\Set N}\mathds{1}\big(|\phi_{k,n}\rangle= |\psi_{k,n}\rangle \big)|\phi_{k,n}\rangle\langle \psi_{k,n} |.
\label{eqn:Rhoall-R}
\end{align}
Equation \eqref{eqn:Rhoall-R} shows that after performing $\Set T$ on all qubit pairs, the state of all qubit pairs is a composition of product states and therefore separable.

Define the function $e(\Set S, \Set O, \Set D)$ as the worst-case squared estimation error for states in the set $\Set S$ by using the measurement operator $\Set O$ and the estimator $\Set D$, i.e.
\begin{align}
e(\Set S, \Set O, \Set D)
= \max_{\V{\rho}_{\mathrm{all}}\in{\Set S}}\sum_{\SM}\mathbb{E}_{R}\big[(\check{F}-\bar{F})^2\big|\SM, \Set O, \Set D\big].
\label{eqn:fSOD}
\end{align}
The objective functions of $\Ps$-\ref{prob:1} and $\Ps$-\ref{prob:sp} can be rewritten as
\begin{align}
\underset{\Set O, \Set D}{\mathrm{minimize}}\, e({\Set S}_{\mathrm{arb}}, \Set O, \Set D), \quad 
\underset{\Set O, \Set D}{\mathrm{minimize}}\, e(\Set S_{\mathrm{sp}}, \Set O, \Set D),
\end{align}
respectively.
 
Denote $\Set D^*$ as the estimator used with the measurement operation $\Set O^*$ to achieve the minimum mean squared error.
The following inequalities involving the objective functions of $\Ps$-\ref{prob:1} and $\Ps$-\ref{prob:sp} can be derived:
\begin{subequations}
\begin{align}
\underset{\Set O, \Set D}{\mathrm{minimize}}\, &e({\Set S}_{\mathrm{arb}}, \Set O, \Set D)\nonumber\\
\leq &e({\Set S}_{\mathrm{arb}}, \hat{\Set O}^*, \Set D^*) \label{eqn:arb-le-sp-a}\\
 \leq &e(\Set S_{\mathrm{sp}}, {\Set O}^*, \Set D^*) \label{eqn:arb-le-sp-b}\\
= & \underset{\Set O, \Set D}{\mathrm{minimize}}\,e(\Set S_{\mathrm{sp}}, \Set O, \Set D), 
\label{eqn:arb-le-sp-c}
\end{align} \label{eqn:arb-le-sp}
\end{subequations}
where \eqref{eqn:arb-le-sp-b} holds because rotation $\Set T$ converts an arbitrary state to a separable state,
and \eqref{eqn:arb-le-sp-c} holds because ${\Set O}^*, \Set D^*$ is the optimal solution to $\Ps$-\ref{prob:sp}.

Moreover, because ${\Set S}_{\mathrm{arb}} \supset \Set S_{\mathrm{sp}}$,
\begin{align}
\underset{\Set O, \Set D}{\mathrm{minimize}}\,e({\Set S}_{\mathrm{arb}}, \Set O, \Set D) 
\geq \underset{\Set O, \Set D}{\mathrm{minimize}}\, e(\Set S_{\mathrm{sp}}, \Set O, \Set D).
\label{eqn:arb-ge-sp}
\end{align}
According to \eqref{eqn:arb-le-sp} and \eqref{eqn:arb-ge-sp}, one can get that
\begin{align}
\begin{split}
\underset{\Set O, \Set D}{\mathrm{minimize}}\, e({\Set S}_{\mathrm{arb}}, \Set O, \Set D)
&= e({\Set S}_{\mathrm{arb}}, \hat{\Set O}^*, \Set D^*)
\\
&= \underset{\Set O, \Set D}{\mathrm{minimize}}\, e(\Set S_{\mathrm{sp}}, \Set O, \Set D).
\end{split}
\label{eqn:arb-eq-sp}
\end{align}
The first equality in \eqref{eqn:arb-eq-sp} shows that $\hat{\Set O}^*$ is optimal for $\Ps$-\ref{prob:1}.
This completes the proof of \thref{lem:sp2arb}.
\end{proof}

\subsection{Sufficiency of considering independent noise}
We next show that to solve $\Ps$-\ref{prob:sp}, it is sufficient to consider a problem in which there is no correlation between different qubit pairs.
This problem is given by $\Ps$-\ref{prob:id}. 
\begin{ProblemA}[Error minimization for independent states] \label{prob:id}
\begin{subequations}
\begin{align}
 \underset{\Set O, \Set D}{\mathrm{minimize}}\;\;&
 \max_{\V{\rho}_{\mathrm{all}}\in{\Set S}_{\mathrm{id}}}\sum_{\SM}\mathbb{E}_{R}\big[(\check{F}-\bar{f}\,)^2\big] 
\label{eqn:minvar_id}\\
\text{subject to}\;\;
&\sum_{\SM}\mathbb{E}_{R}\big[\check{F}-\bar{f}\,\big]=0 ,\;\; \forall \V{\rho}_{\mathrm{all}}\in{\Set S}_{\mathrm{id}}\label{eqn:unbiased_id}
\end{align}\label{eqn:P3}
\end{subequations}
\end{ProblemA}
where
${\Set S}_{\mathrm{id}}$ denotes the set of states $\V{\rho}_{\mathrm{all}}$ in which the state $\V{\rho}_{n}$
of every qubit pair $n$, $n\in \Set N$, is independent, i.e., $\V{\rho}_{\mathrm{all}}$ can be expressed as
\begin{align}
\V{\rho}_{\mathrm{all}} = \otimes^{n\in\Set N} \V{\rho}_{n}.\label{eqn:rhoall-id}
\end{align}
In the case of \eqref{eqn:rhoall-id}, the average fidelity of the unsampled pairs, $\bar{f}$, is deterministic for each sample set $\SM$.
Thus, lowercase letter $\bar{f}$ is used in \eqref{eqn:P3}.

The following lemma shows that if a measurement operation $\Set O^*$ is optimal in scenarios with independent states,
it is also optimal in scenarios with classically correlated states.

\begin{LemA}[Equivalence of $\Ps$-\ref{prob:sp} and $\Ps$-\ref{prob:id}] \thlabel{lem:id2sp}
If a measurement operation $\Set O^*$ is optimal in $\Ps$-\ref{prob:id}, 
it is also optimal in $\Ps$-\ref{prob:sp}.
\end{LemA}
\begin{proof}
A separable state $\V{\rho}_{\mathrm{all}}\in \Set{S}_{\mathrm{sp}}$ can be expressed as
\begin{align}
\V{\rho}_{\mathrm{all}} = \sum_{k\in \Set K}p_k \otimes^{n\in\Set N} \V{\rho}_{k,n}, 
\label{eqn:Rall-sep}
\end{align}
where $p_k$ is the probability of the ensemble in case $k$, with
$
p_k\geq 0$,  $\sum_{k\in \Set K}p_k=1$, and
$\V{\rho}_{k,n}$ is the density matrix of qubit pair $n$ in case $k$.
Denote 
\begin{align}
\V{\rho}^{(k)}_{\mathrm{all}} = \otimes^{n\in\Set N} \V{\rho}_{k,n},
\end{align}
 and
denote the fidelity of $\V{\rho}_{k,n}$ by $f_{k,n}$.
Then $\V{\rho}^{(k)}_{\mathrm{all}} \in \Set{S}_{\mathrm{id}}$, $\forall k \in \Set K$. 

Denote $\check{f}(r)$ as the estimated fidelity given the measurement outcome $r$, 
and denote $\bar{f}(\V{\rho}_{\mathrm{all}},\Set{M}, r)$ as the average fidelity of the unsampled qubit pairs given the state $\V{\rho}_{\mathrm{all}}$, sample set $\SM$ and measurement outcome $r$.
Given \eqref{eqn:Rall-sep}, the following inequalities hold for every estimation protocol $\{\Set O$, $\Set D\}$ and every state $\V{\rho}_{\mathrm{all}}\in \Set{S}_{\mathrm{sp}}$. 
\begin{subequations} 
\begin{align}
&\sum_{\SM}\mathbb{E}_{R}\big[(\check{F}-\bar{F})^2\big]\nonumber \\
&=  \sum_{\SM }\sum_{r}\mathrm{Pr}(r|\SM)
\big(\check{f}(r) -\bar{f}(\V{\rho}_{\mathrm{all}},\Set{M}, r)\big)^2\label{eqn:sp2id-a}\\
&= \sum_{\SM }\sum_{r }\mathrm{Pr}(r|\SM)\big(\check{f}(r) -\nonumber\\
&\hspace{13mm} \sum_{k }
\mathrm{Pr}(\V{\rho}^{(k)}_{\mathrm{all}}| \SM, r) \bar{f}(\V{\rho}^{(k)}_{\mathrm{all}},\Set{M}, r) \big)^2 \label{eqn:sp2id-b}\\
&\leq  \sum_{\SM }\sum_{r }\sum_{k}
\mathrm{Pr}(r|\SM)\mathrm{Pr}(\V{\rho}^{(k)}_{\mathrm{all}}| \SM, r)\nonumber\\
&\hspace{28.5mm} 
\big(\check{f}(r) -\bar{f}(\V{\rho}^{(k)}_{\mathrm{all}},\Set{M}, r) \big)^2 \label{eqn:sp2id-c}\\
&= \sum_{k }p_k\sum_{\SM }\sum_{r }
\mathrm{Pr}( r|\V{\rho}^{(k)}_{\mathrm{all}},\SM)\nonumber\\
&\hspace{28.5mm} 
\big(\check{f}(r) - \bar{f}(\V{\rho}^{(k)}_{\mathrm{all}},\Set{M}, r)\big)^2 \label{eqn:sp2id-d}\\
&= \sum_{k }p_k\sum_{\SM}\mathbb{E}_{R}\big[(\check{F}-\bar{f})^2\big|\V{\rho}^{(k)}_{\mathrm{all}}\big], \label{eqn:sp2id-e}\\
&\leq  \max_{\V{\rho}_{\mathrm{all}}\in{\Set S}_{\mathrm{id}}}\sum_{\SM}\mathbb{E}_{R}\big[(\check{F}-\bar{f})^2\big], \label{eqn:sp2id-f}
\end{align}\label{eqn:sp2id}
\end{subequations}
where \eqref{eqn:sp2id-c} holds because $x^2$ is a convex function and
\begin{align}
\begin{split}
\mathrm{Pr}(\V{\rho}^{(k)}_{\mathrm{all}}| \SM, r)&\geq 0, \quad \forall k \in \Set K,\\
\sum_{k \in \Set K}\mathrm{Pr}(\V{\rho}^{(k)}_{\mathrm{all}}| \SM, r) &=1,
\end{split}
\end{align} 
\eqref{eqn:sp2id-d} holds because according to Bayes' theorem
\begin{align}
\mathrm{Pr}(r|\SM)\mathrm{Pr}(\V{\rho}^{(k)}_{\mathrm{all}}| \SM, r) &=\mathrm{Pr}(\V{\rho}^{(k)}_{\mathrm{all}}|\SM)\mathrm{Pr}(r|\V{\rho}^{(k)}_{\mathrm{all}},\SM)
\nonumber\\
&=p_k \mathrm{Pr}(\SM, r|\V{\rho}^{(k)}_{\mathrm{all}}),
\end{align}
and \eqref{eqn:sp2id-f} holds because $\V{\rho}^{(k)}_{\mathrm{all}} \in \Set{S}_{\mathrm{id}}$, $p_k\geq 0$,  $\forall k \in \Set K$, and $\sum_{k}p_k=1$.

Recall the worst-case squared estimation error function $e(\Set S, \Set O, \Set D)$ defined in \eqref{eqn:fSOD}, and
denote $\Set D^*$ as the estimator used with the measurement operation $\Set O^*$ to achieve the minimum mean squared error.
In this case,
the following inequalities involving the objective functions of $\Ps$-\ref{prob:sp} and $\Ps$-\ref{prob:id} can be derived:
\begin{subequations}
\begin{align}
\underset{\Set O, \Set D}{\mathrm{minimize}}\, &e(\Set S_{\mathrm{sp}}, \Set O, \Set D)\nonumber\\
 \leq &e(\Set S_{\mathrm{sp}}, {\Set O}^*, \Set D^*) \label{eqn:sp-le-id-a}\\
  \leq &e(\Set S_{\mathrm{id}}, {\Set O}^*, \Set D^*) \label{eqn:sp-le-id-b}\\
= & \underset{\Set O, \Set D}{\mathrm{minimize}}\,e(\Set S_{\mathrm{id}}, \Set O, \Set D), 
\label{eqn:sp-le-id-c}
\end{align} \label{eqn:sp-le-id}
\end{subequations}
where \eqref{eqn:sp-le-id-b} is true because that \eqref{eqn:sp2id}  holds for every estimation protocol $\{\Set O$, $\Set D\}$ and every state $\V{\rho}_{\mathrm{all}}\in \Set{S}_{\mathrm{sp}}$, and
\eqref{eqn:sp-le-id-c} holds because $\{\Set O^*, \Set D^*\}$ is the optimal solution to $\Ps$-\ref{prob:id}.

Moreover, because $\Set S_{\mathrm{sp}} \supset \Set S_{\mathrm{id}}$,
\begin{align}
\underset{\Set O, \Set D}{\mathrm{minimize}}\,e(\Set S_{\mathrm{sp}}, \Set O, \Set D) 
\geq \underset{\Set O, \Set D}{\mathrm{minimize}}\, e(\Set S_{\mathrm{id}}, \Set O, \Set D).
\label{eqn:sp-ge-id}
\end{align}

According to \eqref{eqn:sp-le-id} and \eqref{eqn:sp-ge-id},
\begin{align}
\begin{split}
\underset{\Set O, \Set D}{\mathrm{minimize}}\, e(\Set S_{\mathrm{sp}}, \Set O, \Set D)
&= e(\Set S_{\mathrm{sp}}, \Set O^*, \Set D^*)
\\
&= \underset{\Set O, \Set D}{\mathrm{minimize}}\, e(\Set S_{\mathrm{id}}, \Set O, \Set D)
\end{split}
\label{eqn:sp-eq-id}
\end{align}
The first equality in \eqref{eqn:sp-eq-id} shows that ${\Set O}^*$ is optimal for $\Ps$-\ref{prob:sp}.
This completes the proof of \thref{lem:id2sp}.
\end{proof}

\subsection{Sufficiency of considering the sampled qubit pairs}
In a third step, we show that in the case of independent noise, to minimize the estimation error \ac{w.r.t.} the average fidelity of the unsampled qubit pairs,
it is sufficient to minimize the corresponding value of the sampled qubit pairs.
To this end, transform $\Ps$-\ref{prob:id} into $\Ps$-\ref{prob:id2}.
\begin{ProblemA}[Error minimization for sampled qubit pairs] \label{prob:id2}
\begin{subequations}
\begin{align}
 \underset{\Set O, \Set D}{\mathrm{minimize}}\;\;& \max_{\V{\rho}_{\mathrm{all}}\in{\Set S}_{\mathrm{id}}(\V{f}_{\mathrm{all}})}
 \sum_{\SM}\mathbb{E}_{R}\big[(\check{F}-\bar{f}_{\SM})^2\big] 
\label{eqn:minvar_id2}\\
\text{subject to}\;\;
&\mathbb{E}_{R}\big[\check{F}-\bar{f}_{\SM}\big| \V{\rho}_{\SM} \big]=0,\nonumber\\
&\qquad \forall \V{\rho}_{\mathrm{all}}\in{\Set S}_{\mathrm{id}}(\V{f}_{\mathrm{all}}), \; \SM \subset \Set N \label{eqn:unbiased_id2}
\end{align}
\end{subequations}
where $\V{\rho}_{\SM} $ is the state of all sampled qubit pairs, and ${\Set S}_{\mathrm{id}}(\V{f}_{\mathrm{all}})$ denotes the set of $N$ qubit pair states with independent noise and fidelity composition $\V{f}_{\mathrm{all}}=\{f_n,n\in \Set N\}$, i.e.,
\begin{align}
\V{\rho}_{\mathrm{all}} = \otimes^{n\in\Set N} \V{\rho}_{n}, 
\end{align}
in which $\V{\rho}_{n}$ is the state of the $n$-th qubit pair, 
\begin{align}
f_n= \langle\Psi^-|\V{\rho}_n |\Psi^-\rangle, \quad n\in \Set N,
\end{align}
is the fidelity of the $n$-th qubit pair, and
\begin{align}
\bar{f}_{\SM}= \frac{1}{M}\sum_{n\in \SM}f_n
\end{align}
is the average fidelity of the sampled qubit pairs.
\end{ProblemA}

The following lemma shows that a measurement operation $\Set O^*$ which is optimal for $\Ps$-\ref{prob:id2} is also optimal for $\Ps$-\ref{prob:id}.

\begin{LemA}[Equivalence of $\Ps$-\ref{prob:id} and $\Ps$-\ref{prob:id2}] \thlabel{lem:id-sample}
If a measurement operation $\Set O^*$ is optimal in $\Ps$-\ref{prob:id2} for all compositions of fidelity, i.e.,
$\forall\, {\Set S}_{\mathrm{id}}(\V{f}_{\mathrm{all}})\subset{\Set S}_{\mathrm{id}}$, 
it is also optimal in $\Ps$-\ref{prob:id}.
\end{LemA}
\begin{proof}
Because the sample set $\SM$ is drawn completely at random, 
\begin{align}
\sum_{\SM}\bar{f} - \bar{f}_{\Set{M}} =0,\quad \forall \V{\rho}_{\mathrm{all}} \in \Set{S}_{\mathrm{id}}.
\label{eqn:f-f-0}
\end{align}
According to \eqref{eqn:f-f-0}, \eqref{eqn:unbiased_id} is equivalent to
\begin{align}
\sum_{\SM}\mathbb{E}_{R}\big[\check{F}-\bar{f}_{\Set{M}}\big]=0 ,\quad \forall \V{\rho}_{\mathrm{all}}\in{\Set S}_{\mathrm{id}}.\label{eqn:unbiased_id_v2}
\end{align}
\eqref{eqn:unbiased_id_v2} indicates that any estimation protocol $\{\Set O, \Set D\}$ that satisfies \eqref{eqn:unbiased_id2} for all $ {\Set S}_{\mathrm{id}}(\V{f}_{\mathrm{all}})\subset {\Set S}_{\mathrm{id}}$ and $\SM \subset \Set N$ also satisfies \eqref{eqn:unbiased_id}. 
In the following, we prove the converse statement by contradiction.

Suppose there is an estimation protocol $\{\Set O, \Set D\}$ that satisfies \eqref{eqn:unbiased_id_v2} but does not satisfy \eqref{eqn:unbiased_id2}. 
In this case, there must exist some $\V{\rho}_{\SM}$ such that
\begin{align}
\mathbb{E}_{R}\big[\check{F}-\bar{f}_{\SM}\big| \V{\rho}_{\SM}\big]\neq 0.\label{eqn:biased-id2}
\end{align}
Denote $m\in \Set M$ as one of the sampled qubit pairs, and 
denote $l$ as the number of sampled qubit pairs whose state is the same as that of the $m$-th pair, i.e.
\begin{align}
\sum_{n \in \SM}  \mathds{1}(\V{\rho}_n = \V{\rho}_m) = l ,
\end{align}
and denote the state $\V{\rho}_{\SM}$ that satisfies \eqref{eqn:biased-id2} by $\V{\rho}^{(l)}$.
Because $m\in \Set M$, it is clear that $l\geq 1$.

Subsequently, we consider the state
\begin{align}
\V{\rho}_{\mathrm{all}} = \V{\rho}^{(l)} \otimes (\otimes^{N-M} \V{\rho}_m).\label{eqn:rho-N+1}
\end{align}
With the state defined in \eqref{eqn:rho-N+1}, the sampled set $\SM$ has at least $l$ number of qubit pairs in the state $\V{\rho}_m$.
When the number of sampled qubit pairs in state $\V{\rho}_m$ is equal to $l$, the joint state of the sampled qubit pairs is equal to $\V{\rho}^{(l)}$.
Thus, according to \eqref{eqn:unbiased_id_v2} and \eqref{eqn:biased-id2}, there must exist a sample set $\tSM$
such that 
\begin{align}
\begin{split}
\mathbb{E}_{R}\big[\check{F}-\bar{f}_{\tSM}\big|\V{\rho}_{\tSM} \big]&\neq 0,\quad{\mbox{and}}\\
\sum_{n \in \tSM}  \mathds{1}(\V{\rho}_n = \V{\rho}_m) &=\tilde{l} \geq l+1.
\end{split}\label{eqn:induction_result}
\end{align}

Denote the state $\V{\rho}_{\tSM}$ by $\V{\rho}^{(\tilde{l})}$.
Next, consider the state
\begin{align}
\V{\rho}_{\mathrm{all}} = \V{\rho}^{(\tilde{l})} \otimes (\otimes^{N-M} \V{\rho}_m)\label{eqn:rho-N+1}
\end{align}
and repeat the analysis above.  
Because $\tilde{l} \geq l+1$, by repeating this analysis for at most $M-l$ times, it can be obtained that state 
$\V{\rho}^{(M)}=\otimes^M \V{\rho}_m$ does not satisfy the unbiased constraint, i.e., 
\begin{align}
\mathbb{E}_{R}\big[\check{F}-\bar{f}_{\tSM}\big|\V{\rho}^{(M)} \big]&\neq 0.\label{eqn:biased_id2_v2}
\end{align}
However, according to \eqref{eqn:biased_id2_v2}, \eqref{eqn:unbiased_id_v2} does not hold for the state $\V{\rho}_{\mathrm{all}}=\otimes^N \V{\rho}_m$.
This contradiction shows that any estimation protocol $\{\Set O, \Set D\}$ that satisfies \eqref{eqn:unbiased_id} also satisfies \eqref{eqn:unbiased_id2}.

Using the strengthened unbiased constraint \eqref{eqn:unbiased_id2}, the estimation error in $\Ps$-\ref{prob:id}
can be decomposed as follows:
\begin{subequations}
\begin{align}
&\sum_{\SM}\mathbb{E}_{R}\big[(\check{F}-\bar{f})^2\big]\nonumber \\
&=
\sum_{\SM}\mathbb{E}_{R}\big[(\check{F}-\bar{f}_{\SM} +\bar{f}_{\SM} -\bar{f})^2\big]\label{eqn:decompose-a}\\
&=
\sum_{\SM}\mathbb{E}_{R}\big[(\check{F}-\bar{f}_{\SM})^2\big]
+2 \sum_{\SM}\mathbb{E}_{R}\big[(\check{F}-\bar{f}_{\SM})(\bar{f}_{\SM} -\bar{f})\big]\nonumber\\
 &\hspace{4mm}+\sum_{\SM}\mathbb{E}_{R}\big[(\bar{f}_{\SM} -\bar{f})^2\big]\label{eqn:decompose-b}\\
 &=
\sum_{\SM}\mathbb{E}_{R}\big[(\check{F}-\bar{f}_{\SM})^2\big]
+2 \sum_{\SM}\mathbb{E}_{R}\big[\check{F}-\bar{f}_{\SM}\big| \V{\rho}_{\SM}\big](\bar{f}_{\SM} -\bar{f})\nonumber\\
 &\hspace{4mm}+\sum_{\SM}(\bar{f}_{\SM} -\bar{f})^2\label{eqn:decompose-c}\\
 &=\sum_{\SM}\mathbb{E}_{R}\big[(\check{F}-\bar{f}_{\SM})^2\big]+\sum_{\SM}(\bar{f}_{\SM} -\bar{f})^2,\label{eqn:decompose-d}
\end{align}\label{eqn:decompose}
\end{subequations} 
where \eqref{eqn:decompose-c} holds because in cases with independent noise, 
the average fidelity of sampled qubit pairs, $\bar{f}_{\SM}$, and that of unsampled qubit pairs, $\bar{f}$, are deterministic given each sample set $\SM$;
\eqref{eqn:decompose-d} is true due to \eqref{eqn:unbiased_id2}.

\eqref{eqn:decompose} decomposes the estimation error into two parts, i.e., the estimation error \ac{w.r.t.} the average fidelity of the sampled qubit pairs
\begin{align}
\sum_{\SM}\mathbb{E}_{R}\big[(\check{F}-\bar{f}_{\SM} )^2\big],\label{eqn:error-1}
\end{align}
and the deviation between the average fidelity of sampled and unsampled qubit pairs
\begin{align}
\sum_{\SM}(\bar{f}_{\SM} -\bar{f})^2.\label{eqn:error-2}
\end{align}
The value of \eqref{eqn:error-2} is determined by the fidelity composition of all qubit pairs, i.e., 
\begin{align}
\V{f}_{\mathrm{all}} =\{f_n,n\in \Set N\},
\end{align}
and not affected by the estimation protocol $\{\Set O$, $\Set D\}$.
Therefore, for every fidelity composition $\V{f}_{\mathrm{all}}$, minimizing the estimation error
\begin{align}
 \underset{\Set O, \Set D}{\mathrm{minimize}}
 \max_{\V{\rho}_{\mathrm{all}}\in{\Set S}_{\mathrm{id}}(\V{f}_{\mathrm{all}})}\sum_{\SM}\mathbb{E}_{R}\big[(\check{F}-\bar{f})^2\big]
 \label{eqn:minvar_id_f}
 \end{align}
 is equivalent to 
\begin{align}
 \underset{\Set O, \Set D}{\mathrm{minimize}}
 \max_{\V{\rho}_{\mathrm{all}}\in{\Set S}_{\mathrm{id}}(\V{f}_{\mathrm{all}})}\sum_{\SM}\mathbb{E}_{R}\big[(\check{F}-\bar{f}_{\SM})^2\big].
 \end{align} 
Hence, if a protocol $\{\Set O^*$, $\Set D^*\}$ is optimal in $\Ps$-\ref{prob:id2} for all ${\Set S}_{\mathrm{id}}(\V{f}_{\mathrm{all}})\subset{\Set S}_{\mathrm{id}}$, 
it is also the optimal solution to \eqref{eqn:minvar_id_f} for all ${\Set S}_{\mathrm{id}}(\V{f}_{\mathrm{all}})\subset{\Set S}_{\mathrm{id}}$.
Furthermore, because
\begin{align}
{\Set S}_{\mathrm{id}} = \bigcup_{\V{f}_{\mathrm{all}}}{\Set S}_{\mathrm{id}}(\V{f}_{\mathrm{all}}),
\end{align}
protocol $\{\Set O^*$, $\Set D^*\}$ minimizes
\begin{align}
 \max_{\V{\rho}_{\mathrm{all}}\in{\Set S}_{\mathrm{id}}}\sum_{\SM}\mathbb{E}_{R}\big[(\check{F}-\bar{f}\,)^2\big].
 \end{align}
Therefore, $\{\Set O^*$, $\Set D^*\}$ is also an optimal solution to $\Ps$-\ref{prob:id}.
This completes the proof of \thref{lem:id-sample}. 
\end{proof}

The following theorem summarizes the results of this section.
\begin{ThmA}[Generality of optimality with independent noise] \thlabel{thm:gen-opt-id}
If a measurement operation $\Set O^*$ is optimal in $\Ps$-\ref{prob:id2} for all
${\Set S}_{\mathrm{id}}(\V{f}_{\mathrm{all}})\subset{\Set S}_{\mathrm{id}}$,
then a composite measurement operation $\hat{\Set O}^* = \Set O^* \circ \Set T$ is optimal in $\Ps$-\ref{prob:1},
where the operation $\Set T$ is defined as in \thref{def:T}.
\end{ThmA} 
\begin{proof}
The theorem is a direct consequence of Lemmas~\ref{lem:sp2arb}, \ref{lem:id2sp}, and \ref{lem:id-sample}.
\end{proof}

\section{Operation Construction}
\label{sec:construct}
In this section, we will first derive a lower bound of the estimation error,
then based on the conditions for achieveing this lower bound, construct an optimal fidelity estimation protocol.
\subsection{Lower bound of the estimation error} 
\label{subsec:ErrorLB}

We will first further simplify $\Ps$-\ref{prob:id2}. (\thref{lem:id-w}), 
then characterize the limit of separable operators (\thref{lem:sep-ent}) 
and finally determine the lowest feasible bound of the estimation error (\thref{thm:min}).  

The state of a qubit pair $\V{\rho}_{n}\in \Set H^{4\times 4}$ has several parameters other than fidelity.
This property complicates the Fisher information analysis of $\Ps$-\ref{prob:id2}.
To overcome this challenge, we further simplify $\Ps$-\ref{prob:id2} to an equivalent one, $\Ps$-\ref{prob:n}, in which the fidelity composition $\V{f}_{\mathrm{all}}$ is sufficient to parameterize the state of the qubit pairs.
\begin{ProblemA}[Error minimization for Werner states] \label{prob:n}
Given that $\V{\rho}_{\mathrm{all}} = \otimes^{n\in \Set N }\V{\sigma}_n$,
\begin{align}
\underset{\Set O, \Set D}{\mathrm{minimize}}\quad&\sum_{\SM}\mathbb{E}_{R}\big[(\check{F}-\bar{f}_{\SM})^2 \big] \label{eqn:minvar_2}\\
\text{subject to}\quad
&\mathbb{E}_{R}\big[\check{F}-\bar{f}_{\SM}\big| \SM\big]=0,\; \forall \SM \subset \Set N, \label{eqn:unbiased_2}
\end{align}
where $\V{\sigma}_n$ is the Werner state with fidelity $f_n$, i.e., 
\begin{align}
\V{\sigma}_n = f_n|\Psi^-\rangle\langle\Psi^-| + \frac{1-f_n}{3}(\mathbb{I}_4-|\Psi^-\rangle\langle\Psi^-|).\label{eqn:sigma}
\end{align}
\end{ProblemA}

\begin{LemA}[Equivalence of $\Ps$-\ref{prob:id2} and $\Ps$-\ref{prob:n}]
\thlabel{lem:id-w}
With the optimal estimation protocol $\{\Set O, \Set D\}$, the objective functions of $\Ps$-\ref{prob:id2} and $\Ps$-\ref{prob:n} have the same value.
\end{LemA}

\begin{proof}

Recall the function of the worst-case squared estimation error $e(\Set S, \Set O, \Set D)$ defined in \eqref{eqn:fSOD}, and 
denote $\Set{S}_{\mathrm{w}}(\V{f}_{\mathrm{all}}) =\{\otimes^{n\in \Set N }\V{\sigma}_n\}$.
The objective functions of $\Ps$-\ref{prob:id2} and $\Ps$-\ref{prob:n} can be rewritten as
\begin{align}
\begin{split}
&\underset{\Set O, \Set D}{\mathrm{minimize}}\, e({\Set S}_{\mathrm{id}}(\V{f}_{\mathrm{all}}), \Set O, \Set D), \mbox{ and }\\
&\underset{\Set O, \Set D}{\mathrm{minimize}}\, e(\Set S_{\mathrm{w}}(\V{f}_{\mathrm{all}}), \Set O, \Set D),
\end{split}
\end{align}
respectively.

Denote $\{\Set O^*, \Set D^*\}$ as the optimal solution to $\Ps$-\ref{prob:n},
and denote the random bilateral rotation operation proposed in \cite{BenBraPopSchSmoWoo:J96} by $\Set B$.
Define the composition operation $\hat{\Set O}^* = \Set O^* \circ (\otimes^M \Set B)$.
In this case,  we have that
\begin{subequations}
\begin{align}
\underset{\Set O, \Set D}{\mathrm{minimize}}\, &e({\Set S}_{\mathrm{id}}(\V{f}_{\mathrm{all}}), \Set O, \Set D)\nonumber\\
\leq &e({\Set S}_{\mathrm{id}}(\V{f}_{\mathrm{all}}), \hat{\Set O}^*, \Set D^*) \label{eqn:id-le-w-a}\\
 \leq &e({\Set S}_{\mathrm{w}}(\V{f}_{\mathrm{all}}), {\Set O}^*, \Set D^*) \label{eqn:id-le-w-b}\\
= & \underset{\Set O, \Set D}{\mathrm{minimize}}\,e(\Set {\Set S}_{\mathrm{w}}(\V{f}_{\mathrm{all}}), \Set O, \Set D), 
\label{eqn:id-le-w-c}
\end{align} \label{eqn:id-le-w}
\end{subequations}
where \eqref{eqn:id-le-w-b} holds because the operation $\Set B$ transforms a general qubit pair state $\V{\rho}_n$ to a Werner state $\V{\sigma}_n$ with the same fidelity,
and \eqref{eqn:id-le-w-c} holds because $\{\Set O^*, \Set D^*\}$ is the optimal solution to $\Ps$-\ref{prob:sp}.

Moreover, because $\Set{S}_{\mathrm{w}}(\V{f}_{\mathrm{all}})\subset \Set{S}_{\mathrm{id}}(\V{f}_{\mathrm{all}})$,
\begin{align}
\underset{\Set O, \Set D}{\mathrm{minimize}}\,e(\Set S_{\mathrm{w}}(\V{f}_{\mathrm{all}}), \Set O, \Set D) 
\leq \underset{\Set O, \Set D}{\mathrm{minimize}}\, e(\Set S_{\mathrm{id}}(\V{f}_{\mathrm{all}}), \Set O, \Set D).
\label{eqn:w-le-id}
\end{align}

From \eqref{eqn:id-le-w} and \eqref{eqn:w-le-id},
\begin{align}
\underset{\Set O, \Set D}{\mathrm{minimize}}\,e(\Set S_{\mathrm{w}}(\V{f}_{\mathrm{all}}), \Set O, \Set D) 
= \underset{\Set O, \Set D}{\mathrm{minimize}}\, e(\Set S_{\mathrm{id}}(\V{f}_{\mathrm{all}}), \Set O, \Set D),
\end{align}
which proves \thref{lem:id-w}.
\end{proof}

The next lemma characterizes the limit of separable operators when measuring maximally entangled qubit pairs.
\begin{LemA}[Limit of  separable operators]\thlabel{lem:sep-ent}
$|\Phi\rangle$ is a maximally entangled state of a qubit pair. 
$
\M{M} = \M{M}^{\mathrm{(A)}}\otimes \M{M}^{\mathrm{(B)}}
$
is a separable operator, where $\M{M}^{\mathrm{(X)}}\succcurlyeq \M{0}$, $\mathrm{X}\in\{\mathrm{A,B}\}$. 
In this case,
\begin{align}
0\leq \frac{\mathrm{Tr}\big(|\Phi\rangle\langle\Phi|\M{M}^{\mathrm{(AB)}}\big)}
{\mathrm{Tr}\big((\mathbb{I}_4-|\Phi\rangle\langle\Phi|)\M{M}^{\mathrm{(AB)}} \big)}\leq 1.\label{eqn:maxcor}
\end{align}
\end{LemA}
\proof
We define 
\begin{align}
\V{\rho} = \frac{\M{M}^{\mathrm{(AB)}}}{\mathrm{Tr}\big(\M{M}^{\mathrm{(AB)}} \big)}.\label{eqn:rhoMAB}
\end{align} 
Since $\M{M}^{\mathrm{(X)}}\succcurlyeq \M{0}$, $\mathrm{X}\in\{\mathrm{A,B}\}$, $\V{\rho}$ is the density matrix of a separable state of a qubit pair.
For this state, its fidelity \ac{w.r.t.} every maximally entangled state lies in the interval $[0,\frac{1}{2}]$ \cite{HorHor:96,HorHorHor:97}, i.e.,
\begin{align}
\mathrm{Tr}(|\Phi\rangle\langle\Phi|\V{\rho}) \in[0, \frac{1}{2}].\label{eqn:TrAB}
\end{align}
Consequently,
\begin{align}
\frac{\mathrm{Tr}(|\Phi\rangle\langle\Phi|\V{\rho})}{\mathrm{Tr}\big((\mathbb{I}_4-|\Phi\rangle\langle\Phi|)\V{\rho}\big)}=\frac{\mathrm{Tr}(|\Phi\rangle\langle\Phi|\V{\rho})}{1-\mathrm{Tr}(|\Phi\rangle\langle\Phi|\V{\rho})} \in [0,1].\label{eqn:TrAB2}
\end{align}
According to \eqref{eqn:rhoMAB} and \eqref{eqn:TrAB2}, \eqref{eqn:maxcor} is obtained.
This completes the proof of \thref{lem:sep-ent}.
\endproof

\begin{Remark}[The role of inequality \eqref{eqn:maxcor}]\label{remark:local}
To measure the fidelity \ac{w.r.t.} a state $|\Phi\rangle\langle\Phi|$, it is most efficient to use $|\Phi\rangle\langle\Phi|$ and $\mathbb{I}_4-|\Phi\rangle\langle\Phi|$ as measurement operators to ensure that the distribution of the measurement outcome is maximally sensitive to the changes in fidelity.

Since the target state $|\Psi^-\rangle$ is entangled, separable operations cannot realize the measurement operators mentioned above.
To increase their efficiency, separable measurement operators should be designed to best mimic the ideal operator $|\Psi^-\rangle\langle\Psi^-|$.
The inequality \eqref{eqn:maxcor} characterizes the best alignment between separable operators and the target state. 
Consequently, \eqref{eqn:maxcor} is the critical constraint that upper bounds the efficiency of separable operators in fidelity estimation.
~\QEDB
\end{Remark}

Based on the above results, the following lemma characterizes a lower bound for the estimation error in $\Ps$-\ref{prob:id2}.
\begin{LemA}[Lower bound for the estimation error]\thlabel{thm:min}
The mean squared estimation error, i.e., the objective function of $\Ps$-\ref{prob:id2} divided by $N \choose M$, is no less than 
\begin{align}
\sum_{n\in\Set N}\frac{(2f_n+1)(1-f_n)}{2MN}.
\label{eqn:errorbound}
\end{align}
\end{LemA}
\proof

According to \thref{lem:id-w}, the minimum value of the objective function of $\Ps$-\ref{prob:id2} is equal to that of  $\Ps$-\ref{prob:n}.
Therefore, the following analysis lower bounds the objective function of $\Ps$-\ref{prob:n}.

In $\Ps$-\ref{prob:n}, $\V{\rho}_{\mathrm{all}} = \otimes^{n\in \Set N }\V{\sigma}_n$, i.e., the qubit pairs are in independent Werner states.
In this case, for each sampled set $\SM$, the fidelity composition of the sampled qubit pairs,
\begin{align}
\V{f}_{\SM} =\{f_n, n\in \SM\},
\end{align}
is the unknown fixed parameter that determines the distribution of the measurement outcome $R$.
According to the unbiased constraint \eqref{eqn:unbiased_id2} and the fact that $F_{\SM} = \frac{1}{M}\sum_{n\in \SM} f_n$,
\begin{align}
\frac{\partial \mathbb{E}_R\big[\check{F} \big]}{\partial \V{f}_{\SM}} = \frac{\V{1}_M}{M},\label{eqn:partial-E}
\end{align}
where $\V{1}_M$ denotes the $1\times M$ all-ones vector.
In this case, according to the Cram\'er-Rao bound \cite{Cra:B99,Rao:B94},
\begin{align}
\mathbb{E}_{R}\big[(\check{F}-\bar{f}_{\SM})^2\big] 
\geq \frac{\partial \mathbb{E}_R\big[\check{F} \big]}{\partial \V{f}_{\SM}}\big({\M I}_{R}(\V{f}_{\SM})\big)^{-1}\bigg(\frac{\partial \mathbb{E}_R\big[\check{F} \big]}{\partial \V{f}_{\SM}}\bigg)^T,\label{eqn:EtoF}
\end{align}
where ${\M I}_{R}(\V{f}_{\SM})$ is the $M\times M$ Fisher information matrix of the measurement outcome $R$ evaluated at point $\V{f}_{\SM}$.
The elements of ${\M I}_{R}(\V{f}_{\SM})$ are specified as
\begin{align}
I_{n,m}(\V{f}_{\SM}) = \mathbb{E}_{R}\bigg[\frac{\partial}{\partial f_n}\log {\Set P}(R;\V{f}_{\SM}) \Big(\frac{\partial}{\partial f_m}\log {\Set P}(R;\V{f}_{\SM})\Big)^T\bigg],
\label{eqn:Idk}
\end{align}
where $n,m \in \SM$, and ${\Set P}(r;\V{f}_{\SM})$ is the distribution of measurement outcome $r$ given the fidelity composition $\V{f}_{\SM}$.

According to \cite[Prop. 1]{BobMayZak:J87}, the reciprocal of the diagonal elements of a Fisher information matrix lower bounds the inverse of that matrix, i.e.
\begin{align}
\big({\M I}_{R}(\V{f}_{\SM})\big)^{-1} \succcurlyeq \mathrm{diag}\big(I^{-1}_{n,n}(\V{f}_{\SM}), n\in \SM \big),\label{eqn:Fisherinverse}
\end{align} 
where $\succcurlyeq$ denotes the matrix inequality, and $\mathrm{diag}(\cdot)$ denotes the diagonal matrix.
Substituting \eqref{eqn:partial-E} and \eqref{eqn:Fisherinverse} into \eqref{eqn:EtoF}, we get
\begin{align}
\mathbb{E}_{R}\big[(\check{F}-\bar{f}_{\SM})^2\big] 
\geq\sum_{n\in \SM} \frac{I^{-1}_{n,n}(\V{f}_{\SM})}{M^2}.\label{eqn:EtoF-2}
\end{align}

Given that the state of qubit pairs $\V{\rho}_{\mathrm{all}} = \otimes^{n\in \Set N }\V{\sigma}_n$,
and measurement operator 
\begin{align}
\M{M}_{r}= \sum_{k}\otimes^{n\in \Set M} (\M{M}^{\mathrm{(A)}}_{r,n,k}\otimes \M{M}^{\mathrm{(B)}}_{r,n,k}),
\end{align}
the distribution of the measurement outcome can be expressed as the following.
\begin{align}
&{\Set P}(r;\V{f}_{\SM})\nonumber\\
&= \mathrm{Tr}\big(\otimes^{n\in \Set M}\V{\sigma}_n\M{M}_{r}\big)\nonumber\\
&=\sum_{k}\mathrm{Tr}\big(\otimes^{n\in \Set M}(\V{\sigma}_n\M{M}^{\mathrm{(A)}}_{r,n,k}\otimes \M{M}^{\mathrm{(B)}}_{r,n,k})\big)\nonumber\\
&=\sum_{k}\prod_{n\in \Set M}\mathrm{Tr}\big(\V{\sigma}_n\M{M}^{\mathrm{(A)}}_{r,n,k}\otimes \M{M}^{\mathrm{(B)}}_{r,n,k}\big).
\label{eqn:partial-sep}
\end{align}

To characterize the effect of fidelity $f_n$ on the distribution of the measurement outcome, define
\begin{align}
\begin{split}
a_{r,n} &= \sum_{k}\mathrm{Tr}\big(|\Psi^-\rangle\langle\Psi^-|\M{M}^{\mathrm{(A)}}_{r,n,k}\otimes \M{M}^{\mathrm{(B)}}_{r,n,k}\big)\\
&\hspace{0.35mm}\prod_{m\in \Set M \backslash\{n\}}\mathrm{Tr}\big(\V{\sigma}_m\M{M}^{\mathrm{(A)}}_{r,m,k}\otimes \M{M}^{\mathrm{(B)}}_{r,m,k}\big),\\
b_{r,n} &= \sum_{k}\mathrm{Tr}\big((\mathbb{I}_4-|\Psi^-\rangle\langle\Psi^-|)\M{M}^{\mathrm{(A)}}_{r,n,k}\otimes \M{M}^{\mathrm{(B)}}_{r,n,k}\big)\\
&\hspace{0.35mm}\prod_{m\in \Set M \backslash\{n\}}\mathrm{Tr}\big(\V{\sigma}_m\M{M}^{\mathrm{(A)}}_{r,m,k}\otimes \M{M}^{\mathrm{(B)}}_{r,m,k}\big).
\end{split}\label{eqn:trdef}
\end{align}
It is evident that the value of $f_n$ has no effect on $a_{r,n}$ and $b_{r,n}$, i.e.
\begin{align}
\frac{\partial a_{r,n} }{ \partial f_n} = \frac{\partial b_{r,n} }{ \partial f_n} = 0, \; \forall r, n.\label{eqn:partial0}
\end{align}

According to \eqref{eqn:sigma} and \eqref{eqn:partial-sep}, the distribution of the measurement outcome can be rewritten as
\begin{align}
&\Set{P}(r|\V{f}_{\SM})\nonumber\\
&= f_n \sum_{k}\mathrm{Tr}\big(|\Psi^-\rangle\langle\Psi^-|\M{M}^{\mathrm{(A)}}_{r,n,k}\otimes \M{M}^{\mathrm{(B)}}_{r,n,k}\big)\nonumber\\
&\hspace{4mm}\prod_{m\in \Set M \backslash\{n\}}\mathrm{Tr}\big(\V{\sigma}_m\M{M}^{\mathrm{(A)}}_{r,m,k}\otimes \M{M}^{\mathrm{(B)}}_{r,m,k}\big)\nonumber\\
&\hspace{4mm}+ \frac{1-f_{n}}{3}\sum_{k}\mathrm{Tr}\big((\mathbb{I}_4-|\Psi^-\rangle\langle\Psi^-|)\M{M}^{\mathrm{(A)}}_{r,n,k}\otimes \M{M}^{\mathrm{(B)}}_{r,n,k}\big)\nonumber\\
&\hspace{4mm}\prod_{m\in \Set M \backslash\{n\}}\mathrm{Tr}\big(\V{\sigma}_m\M{M}^{\mathrm{(A)}}_{r,m,k}\otimes \M{M}^{\mathrm{(B)}}_{r,m,k}\big)\nonumber\\
&= f_{n} a_{r,n} + \frac{1-f_{n}}{3}b_{r,n}.
\label{eqn:pmd_rt}
\end{align}
According to \eqref{eqn:Idk}, \eqref{eqn:partial0}, and \eqref{eqn:pmd_rt},
\begin{align}
I_{n,n}(\V{f}_{\Set M}) &= \sum_{r}\Set{P}(r|f_{n})\Big(\frac{\partial}{\partial f_{n}}\log {\Set P}(r|f_{n}) \Big)^2\nonumber \\
&= \sum_{r}\frac{\big(a_{r,n}-\frac{1}{3}b_{r,n}\big)^2}{f_na_{r,n}+\frac{1-f_n}{3}b_{r,n}}.
\label{eqn:F_rt}
\end{align}

According to the property of \ac{POVM}, i.e. $\sum_{r}\M{M}_{r} =\mathbb{I}_{4^M}$,
it can be obtained that
\begin{align}
\sum_{r}a_{r,n}&= \mathrm{Tr}\Big(|\Psi^-\rangle\langle\Psi^-|\otimes\big(\otimes^{m\in \Set M\backslash\{n\}}\V{\sigma}_m\big)\sum_{r}\M{M}_{r}\Big)\nonumber\\
& = \mathrm{Tr}\Big(|\Psi^-\rangle\langle\Psi^-|\otimes\big(\otimes^{m\in \Set M\backslash\{n\}}\V{\sigma}_m\big)\Big)\nonumber\\
& = \mathrm{Tr}(|\Psi^-\rangle\langle\Psi^-|)\prod_{m\in \Set M \backslash\{n\}}\mathrm{Tr}(\V{\sigma}_m)\nonumber\\
& = 1.\label{eqn:suma}
\end{align}
Similarly,
\begin{align}
\sum_{r}b_{r,n}
& = \mathrm{Tr}(\mathbb{I}_4-|\Psi^-\rangle\langle\Psi^-|)\prod_{m\in \Set M \backslash\{n\}}\mathrm{Tr}(\V{\sigma}_m)\nonumber\\
& = 3.\label{eqn:sumb}
\end{align}

Applying \thref{lem:sep-ent}, we get
\begin{align}
0&\leq \mathrm{Tr}\big(|\Psi^-\rangle\langle\Psi^-|\M{M}^{\mathrm{(A)}}_{r,n,k}\otimes \M{M}^{\mathrm{(B)}}_{r,n,k}\big) \nonumber\\
& \leq \mathrm{Tr}\big((\mathbb{I}_4-|\Psi^-\rangle\langle\Psi^-|)\M{M}^{\mathrm{(A)}}_{r,n,k}\otimes \M{M}^{\mathrm{(B)}}_{r,n,k}\big),\; \forall r, n, k
\label{eqn:Tr-rnk}
\end{align}
Substituting \eqref{eqn:Tr-rnk} into \eqref{eqn:trdef}, one obtains that
\begin{align}
0\leq a_{r,n} \leq b_{r,n},\quad \forall r, n.\label{eqn:rt_property}
\end{align}

Define
\begin{align}
I_{r,n}(a_{r,n},{b}_{r,n}) = \frac{\big(a_{r,n}-\frac{1}{3}{b}_{r,n}\big)^2}{f_na_{r,n}+\frac{1-f_n}{3}{b}_{r,n}},
\end{align}
then because
\begin{align}
\frac{\partial^2 I_{r,n}}{\partial^2 a_{r,n}} = \frac{6 {b}^2_{r,n}}{\big((1-f_n)b_{r,n} + 3f_na_{r,n}\big)^3}\geq 0,
\label{eqn:convex}
\end{align}
the function $I_{r,n}$ is convex \ac{w.r.t.} $a_{r,n}$. Therefore, according to \eqref{eqn:rt_property},
\begin{align}
&I_{r,n}(a_{r,n},b_{r,n}) \nonumber\\
&\quad\leq \frac{b_{r,n}-a_{r,n}}{b_{r,n}}I_{r,n}(0,b_{r,n}) + \frac{a_{r,n}}{b_{r,n}}I_{r,n}(b_{r,n},b_{r,n})\nonumber\\
&\quad = \frac{1}{3(1-f_n)} (b_{r,n}-a_{r,n}) + \frac{4}{3(2f_n+1)} a_{r,n}.
\label{eqn:convex_relax}
\end{align}

Substituting \eqref{eqn:suma}, \eqref{eqn:sumb}, and \eqref{eqn:convex_relax}
into \eqref{eqn:F_rt}, we obtain that 
\begin{align}
&I_{n,n}(\V{f}_{\Set M})\nonumber\\
&= \sum_{r}I_{r,n}(a_{r,n},b_{r,n})\nonumber\\
&\leq \frac{1}{3(1-f_n)} \sum_{r} (b_{r,n}-b_{r,n}) + \frac{4}{3(2f_n+1)} \sum_{r} b_{r,n}\nonumber\\
&=\frac{2}{3(1-f_n)} + \frac{4}{3(2f_n+1)}\nonumber\\
&= \frac{2}{(2f_n+1)(1-f_n)}.\label{eqn:Inn-upper}
\end{align}
Substituting \eqref{eqn:Inn-upper} into \eqref{eqn:EtoF-2}, the mean squared estimation error is lower bounded as below.
\begin{align}
&\hspace{-10mm}{N\choose M}^{-1}\sum_{\SM}\mathbb{E}_{R}\big[(\check{F}-\bar{f}_{\SM})^2 \big]\nonumber\\
&\geq {N\choose M}^{-1} \sum_{\SM } \sum_{n\in \SM} \frac{I_{n,n}^{-1}(\V{f}_{\SM})}{M^2}\nonumber\\
&\geq {N\choose M}^{-1}\sum_{n\in \Set N}{N-1\choose M-1} \frac{(2f_{n}+1)(1-f_{n})}{2M^2}\nonumber\\
&=\sum_{n\in\Set N}\frac{(2f_n+1)(1-f_n)}{2MN}.
\label{eqn:error-combine-2}
\end{align}
With \eqref{eqn:error-combine-2}, \eqref{eqn:errorbound} is obtained, which completes the proof of \thref{thm:min}.
\endproof

\subsection{Achieving the minimum estimation error}
\label{subsec:optoperator}
This subsection describes the measurement operation used to achieve the minimum estimation error. 

According to the proof of \thref{thm:min}, a necessary condition for achieving the minimum estimation error is that the inequality in \eqref{eqn:convex_relax} is balanced.
To this end, the projection of a measurement operator onto the target state, i.e., $a_{r,n}$ defined in \eqref{eqn:trdef}, must balance either of the two inequalities in \eqref{eqn:rt_property}.
The estimation protocol is built according to this principle.
Specifically, we consider bilateral local Pauli measurements in the same basis.
When both nodes make measurements in the $u$-basis, $u\in \{x,y,z\}$, the operators corresponding to the asymmetric and symmetric measurement results, $\M{M}_{u,0}$ and $\M{M}_{u,1}$, are respectively given by
\begin{subequations}
\begin{align}
\M{M}_{x,0}&=\ket{+-}\!\bra{+-} + \ket{-+}\!\bra{-+}\nonumber\\ 
&= |\Psi^-\rangle\langle\Psi^-| + |\Phi^-\rangle\langle\Phi^-|, \\
\M{M}_{x,1}&=\ket{++}\!\bra{++} + \ket{--}\!\bra{--}\nonumber\\ 
&= |\Psi^+\rangle\langle\Psi^+| + |\Phi^+\rangle\langle\Phi^+|, \\
\M{M}_{y,0}&=\ket{\times \odot}\!\bra{\times\odot} + \ket{\odot \times}\!\bra{\odot \times}\nonumber\\ 
&= |\Psi^-\rangle\langle\Psi^-| + |\Phi^+\rangle\langle\Phi^+|, \\
\M{M}_{y,1}&=\ket{\times \times}\!\bra{\times\times} + \ket{\odot \odot}\!\bra{\odot \odot}\nonumber\\ 
&= |\Psi^+\rangle\langle\Psi^+| + |\Phi^-\rangle\langle\Phi^-|, \\
\M{M}_{z,0}&=\ket{01}\!\bra{01} + \ket{10}\!\bra{10}\nonumber\\ 
&= |\Psi^-\rangle\langle\Psi^-| + |\Psi^+\rangle\langle\Psi^+|, \\
\M{M}_{z,1}&=\ket{00}\!\bra{00} + \ket{11}\!\bra{11}\nonumber\\ 
&= |\Phi^-\rangle\langle\Phi^-| + |\Phi^+\rangle\langle\Phi^+|, 
\end{align}
\label{eqn:protocloperators}
\end{subequations}
where $\ket{+} = \frac{\ket{0} + \ket{1}}{\sqrt{2}}$, $\ket{-} = \frac{\ket{0} - \ket{1}}{\sqrt{2}}$,
$\ket{\times} = \frac{\ket{0} + i\ket{1}}{\sqrt{2}}$, and $\ket{\odot} = \frac{\ket{0} - i\ket{1}}{\sqrt{2}}$.

In \eqref{eqn:protocloperators}, the operators $\M{M}_{u,0}$, $u\in \{x,y,z\}$ balance the second inequality in \eqref{eqn:rt_property},
and the operators $\M{M}_{u,1}$, $u\in \{x,y,z\}$ balance the first inequality in \eqref{eqn:rt_property}.
The following two lemma confirms the optimality of the proposed measurement operator.

\begin{LemA}[Unbiased estimator]\thlabel{lem:unbia}
Given the measurements made in Protocol~1, the following estimated fidelity
\begin{align}
\check{f}=1-\frac{3\varepsilon_{\SM}}{2}\label{eqn:estfid}
\end{align}
 satisfies \eqref{eqn:unbiased_id2}, where the \ac{QBER} $\varepsilon_{\SM}$ is expressed as
\begin{align}
\varepsilon_{\SM}= \frac{\sum_{n\in\SM} r_{n}}{M},\label{eqn:varM}
\end{align}
in which $r_{n}$ is the measurement result of qubit $n$.
\end{LemA}
\begin{proof}
We define the following four terms:
\begin{align}
\begin{split}
f^{(0)} = \frac{1}{M}\sum_{n\in \SM}\mathrm{Tr}(\V{\rho}_n|\Psi^- \rangle\langle\Psi^- |),\\
f^{(1)} = \frac{1}{M}\sum_{n\in \SM}\mathrm{Tr}(\V{\rho}_n|\Psi^+ \rangle\langle\Psi^+ |),\\
f^{(2)} = \frac{1}{M}\sum_{n\in \SM}\mathrm{Tr}(\V{\rho}_n|\Phi^- \rangle\langle\Phi^- |),\\
f^{(3)} = \frac{1}{M}\sum_{n\in \SM}\mathrm{Tr}(\V{\rho}_n|\Phi^+ \rangle\langle\Phi^+ |).
\end{split}
\label{eqn:F0123}
\end{align}
In this case, $\bar{f}_{\SM} = f^{(0)}$. Moreover, as the Bell states form a basis of $\mathbb{C}^4$,
\begin{align}
\sum_{i=0}^{3}f^{(i)} =  \frac{1}{M}\sum_{n\in \SM}\mathrm{Tr}(\V{\rho}_n)=1.\label{eqn:Fsum}
\end{align}

Given the distribution of $A_n$, i.e.,
$\Pr[A_n=u]=\frac{1}{3}$, $u\in\{x,y,z\}$,
 and the expression of the measurement operators \eqref{eqn:protocloperators}, the following expression can be obtained:
\begin{align}
&\mathbb{E}\bigg[\frac{\sum_{n\in \Set{M}} R_{n}  \mathds{1}( A_n = x)}{M}\bigg|\bar{f}_{\SM}\bigg]\nonumber\\
&= \frac{1}{M}\sum_{n\in \Set{M}} \Pr[A_n=x]\mathbb{E}\big[ R_{n} \big| A_n=x,\bar{f}_{\SM}\,\big]\nonumber\\
&=\frac{1}{3M}\sum_{n\in \SM}\mathrm{Tr}\big(
 \V{\rho}_n (\mathbb{I}_4 - \M{M}_{x,0})\big)
\nonumber\\
&= \frac{1}{3M}\bigg(\sum_{n\in \SM}\mathrm{Tr}\big(
 \V{\rho}_n  |\Psi^+\rangle\langle\Psi^+|\big)
 +\sum_{n\in \SM}\mathrm{Tr}\big(
 \V{\rho}_n  |\Phi^+\rangle\langle\Phi^+|\big)\bigg)\nonumber\\
&= \frac{f^{(1)}+f^{(3)}}{3}.\label{eqn:F02}
\end{align}
Similar to \eqref{eqn:F02}, one can obtain that
\begin{align}
\mathbb{E}\bigg[\frac{\sum_{n\in \Set{M}} R_{n}  \mathds{1}( A_n = y)}{M}\bigg|\bar{f}_{\SM}\bigg]\hspace{-8.8mm}\nonumber\\
&= \frac{1}{3M}\sum_{n\in \SM}\mathrm{Tr}\big(
 \V{\rho}_n  (\mathbb{I}_4 - \M{M}_{y,0})\big)
\nonumber\\
&= \frac{f^{(1)}+f^{(2)}}{3},\label{eqn:F03}\\
\mathbb{E}\bigg[\frac{\sum_{n\in \Set{M}} R_{n}  \mathds{1}( A_n = z)}{M}\bigg|\bar{f}_{\SM}\bigg]\hspace{-8.8mm}\nonumber\\
&= \frac{1}{3M}\sum_{n\in \SM}\mathrm{Tr}\big(
 \V{\rho}_n  (\mathbb{I}_4 - \M{M}_{z,0})\big)
\nonumber\\
&= \frac{f^{(2)}+f^{(3)}}{3}.\label{eqn:F01}
\end{align}
Substituting \eqref{eqn:Fsum}--\eqref{eqn:F01} into \eqref{eqn:estfid}, yields

\begin{align}
&\mathbb{E}\big[\check{F}\big|\bar{f}_{\SM}\,\big]\nonumber\\
&= 1-\frac{3}{2}\mathbb{E}\big[\Set{E}_{\SM}\big|\bar{f}_{\SM}\big]\nonumber\\
&= 1-\frac{3}{2}
\mathbb{E}\bigg[\frac{\sum_{a\in\{x,y,z\}}\sum_{n\in \Set{M}} R_{n}  \mathds{1}\{ A_n = a\}}{M}\bigg|\bar{f}_{\SM}\bigg]\nonumber\\
&=1- \frac{3}{2}\bigg(\sum_{i=1}^3 \frac{2f^{(i)}}{3}\bigg)\nonumber\\
&=f^{(0)} = \bar{f}_{\SM},
\label{eqn:barF}
\end{align}
where $\check{F}$ and ${\Set E}_{\SM}$ are the random variable form of $\check{f}$ and $\varepsilon_{\SM}$,  respectively.
\eqref{eqn:barF} shows that the estimator given in \eqref{eqn:estfid} satisfies \eqref{eqn:unbiased_id2}. 
This completes the proof of \thref{lem:unbia}.
\end{proof}

Next, \thref{lem:opt-id} proves the optimality of the proposed measurement operation in scenarios with independent noise.

\begin{LemA}[Optimality with independent noise]\thlabel{lem:opt-id}
Protocol~1 achieves the minimum estimation error in $\Ps$-\ref{prob:id2}.
\end{LemA}
\begin{proof}
By repeating the derivation from \eqref{eqn:F02} to \eqref{eqn:barF} for qubit pair $n$, the following expression can be obtained
\begin{align}
\mathrm{Pr}(R_n =1) = \frac{2}{3}(1-f_n).
\end{align}
Hence, the variance of $R_n$ is given by
\begin{align}
\mathbb{V}\big[R_n\big]&= \mathrm{Pr}(R_n =1) (1-\mathrm{Pr}(R_n =1) )\nonumber\\
 &= \frac{(2f_n+1)(2-2f_n)}{9}\label{eqn:VRn}
\end{align}

In the case of independent noise, the measurement outcomes $R_n$, $n\in \SM$ on different qubit pairs are independent.
Therefore, according to \eqref{eqn:estfid} and \eqref{eqn:varM}
\begin{align}
\mathbb{V}\big[\check{F}\big| \SM\big]
&= \frac{9}{4M^2}\sum_{n\in\SM} \mathbb{V}\big[R_n\big]\nonumber\\
&=\sum_{n\in\SM}\frac{(2f_n+1)(1-f_n)}{2M^2}.\label{eqn:Fvar-M}
\end{align}

According to Lemma~\ref{lem:unbia}, $\check{F}$ satisfies \eqref{eqn:unbiased_id2}, hence
\begin{align}
\mathbb{E}_{R}\big[(\check{F}-\bar{F}_{\SM})^2\big| \SM \big]= \mathbb{V}\big[\check{F}\big| \SM\big].
\label{eqn:E2V-F}
\end{align}
As the sample set $\SM$ is selected completely at random, \eqref{eqn:Fvar-M} and \eqref{eqn:E2V-F} indicate that
for all ${\Set S}_{\mathrm{id}}(\V{f}_{\mathrm{all}})\subset{\Set S}_{\mathrm{id}}$,
\begin{align}
&{N\choose M}^{-1} \sum_{\SM \subset \Set N}\mathbb{E}_{R}\big[(\check{F}-\bar{F}_{\SM})^2 \big]\nonumber\\ 
& = {N\choose M}^{-1} \sum_{\SM \subset \Set N}\sum_{n\in\SM}\frac{(2f_n+1)(1-f_n)}{2M^2}\nonumber\\
& = {N\choose M}^{-1} \sum_{n\in\Set N}{N-1\choose M-1}\frac{(2f_n+1)(1-f_n)}{2M^2}\nonumber\\
&=\sum_{n\in\Set N}\frac{(2f_n+1)(1-f_n)}{2MN}.\label{eqn:error-upperbound}
\end{align}

According to \thref{thm:min}, the estimation error of $\Ps$-\ref{prob:id2} is no less than \eqref{eqn:error-upperbound}.
This aspect shows that Protocol~1 is optimal in $\Ps$-\ref{prob:id2} for all fidelity compositions ${\Set S}_{\mathrm{id}}(\V{f}_{\mathrm{all}})\subset{\Set S}_{\mathrm{id}}$.
This completes the proof of \thref{lem:opt-id}.
\end{proof}

Finally, \thref{thm:opt} summarizes the results of this section.

\begin{ThmA}[Optimal estimation protocol]\thlabel{thm:opt}
Protocol~1 is optimal for $\Ps$-\ref{prob:1}.
\end{ThmA}
\begin{proof}
Denote the measurement operation and estimator of Protocol~1 by $\Set O^*$ and $\Set D^*$, respectively.
According to \thref{lem:opt-id} and \thref{thm:gen-opt-id}, the estimation protocol $\{\hat{\Set O}^*, \Set D^* \}$ is optimal for $\Ps$-\ref{prob:1}, where the composite measurement operation $\hat{\Set O}^* = \Set O^* \circ \Set T$.

Recall the following property of the operation $\Set T$ given in \eqref{eqn:R-phi}
\begin{align}
\Set{T}(|\phi\rangle\langle \psi | ) &= 
\left\{\begin{array}{ll}
|\phi\rangle\langle \psi | & \mbox{ if } |\phi\rangle = |\psi\rangle \\
0 & \mbox{ if } |\phi\rangle \neq |\psi\rangle
\end{array}\right.,\label{eqn:R-phi-2}
\end{align}
where $|\phi\rangle, |\psi\rangle$ are Bell states  $\{|\Phi^\pm\rangle, |\Psi^\pm\rangle\}$.
Given \eqref{eqn:R-phi-2} and the fact that Bell states form a basis of $\Set{H}^{4}$, the Kraus operators of $\Set{T}$ processing one qubit pair can be expressed as
\begin{align}
\M{T}_{\phi} &=  |\phi\rangle\langle \phi |, \quad \mbox{where}\quad \phi \in \{\Phi^\pm, \Psi^\pm\}.\label{eqn:KrausT}
\end{align}  
In this case, \eqref{eqn:protocloperators} and \eqref{eqn:KrausT},
operation $\Set{T}$ does not change the operators of $\Set O^*$, i.e., 
\begin{align}
\M{M}_{u,r} = \sum_{\phi \in \{\Phi^\pm, \Psi^\pm\}} \M{T}^\dag_{\phi}\M{M}_{u,r}\M{T}_{\phi},\;\; \forall\, u\in\{x,y,z\},\, r\in\{0,1\},
\label{eqn:TMur}
\end{align}
where $^\dag$ is the Hermitian transpose. \eqref{eqn:TMur} shows that the composite measurement operation $\hat{\Set O}^* = \Set O^* \circ \Set T$ is equivalent to $\Set O^*$, i.e., 
\begin{align}
\hat{\Set O}^* = \Set O^* \circ \Set T =\Set O^*.\label{eqn:OT=O_ap}
\end{align}
According to \eqref{eqn:OT=O_ap}, the estimation protocol $\{\Set O^*, \Set D^*\}$ of Protocol~1 is optimal for $\Ps$-\ref{prob:1}. This proves \thref{thm:opt}.
\end{proof}


\begin{thebibliography}{38}%
\makeatletter
\providecommand \@ifxundefined [1]{%
 \@ifx{#1\undefined}
}%
\providecommand \@ifnum [1]{%
 \ifnum #1\expandafter \@firstoftwo
 \else \expandafter \@secondoftwo
 \fi
}%
\providecommand \@ifx [1]{%
 \ifx #1\expandafter \@firstoftwo
 \else \expandafter \@secondoftwo
 \fi
}%
\providecommand \natexlab [1]{#1}%
\providecommand \enquote  [1]{``#1''}%
\providecommand \bibnamefont  [1]{#1}%
\providecommand \bibfnamefont [1]{#1}%
\providecommand \citenamefont [1]{#1}%
\providecommand \href@noop [0]{\@secondoftwo}%
\providecommand \href [0]{\begingroup \@sanitize@url \@href}%
\providecommand \@href[1]{\@@startlink{#1}\@@href}%
\providecommand \@@href[1]{\endgroup#1\@@endlink}%
\providecommand \@sanitize@url [0]{\catcode `\\12\catcode `\$12\catcode
  `\&12\catcode `\#12\catcode `\^12\catcode `\_12\catcode `\%12\relax}%
\providecommand \@@startlink[1]{}%
\providecommand \@@endlink[0]{}%
\providecommand \url  [0]{\begingroup\@sanitize@url \@url }%
\providecommand \@url [1]{\endgroup\@href {#1}{\urlprefix }}%
\providecommand \urlprefix  [0]{URL }%
\providecommand \Eprint [0]{\href }%
\providecommand \doibase [0]{https://doi.org/}%
\providecommand \selectlanguage [0]{\@gobble}%
\providecommand \bibinfo  [0]{\@secondoftwo}%
\providecommand \bibfield  [0]{\@secondoftwo}%
\providecommand \translation [1]{[#1]}%
\providecommand \BibitemOpen [0]{}%
\providecommand \bibitemStop [0]{}%
\providecommand \bibitemNoStop [0]{.\EOS\space}%
\providecommand \EOS [0]{\spacefactor3000\relax}%
\providecommand \BibitemShut  [1]{\csname bibitem#1\endcsname}%
\let\auto@bib@innerbib\@empty
\bibitem [{\citenamefont {Wehner}\ \emph {et~al.}(2018)\citenamefont {Wehner},
  \citenamefont {Elkouss},\ and\ \citenamefont {Hanson}}]{WehElkHan:J18}%
  \BibitemOpen
  \bibfield  {author} {\bibinfo {author} {\bibfnamefont {S.}~\bibnamefont
  {Wehner}}, \bibinfo {author} {\bibfnamefont {D.}~\bibnamefont {Elkouss}},\
  and\ \bibinfo {author} {\bibfnamefont {R.}~\bibnamefont {Hanson}},\
  }\bibfield  {title} {\bibinfo {title} {Quantum internet: A vision for the
  road ahead},\ }\href@noop {} {\bibfield  {journal} {\bibinfo  {journal}
  {Science}\ }\textbf {\bibinfo {volume} {362}} (\bibinfo {year}
  {2018})}\BibitemShut {NoStop}%
\bibitem [{\citenamefont {Kimble}(2008)}]{Kim:J08}%
  \BibitemOpen
  \bibfield  {author} {\bibinfo {author} {\bibfnamefont {H.~J.}\ \bibnamefont
  {Kimble}},\ }\bibfield  {title} {\bibinfo {title} {The quantum internet},\
  }\href@noop {} {\bibfield  {journal} {\bibinfo  {journal} {Nature}\ }\textbf
  {\bibinfo {volume} {453}},\ \bibinfo {pages} {1023} (\bibinfo {year}
  {2008})}\BibitemShut {NoStop}%
\bibitem [{\citenamefont {Castelvecchi}(2018)}]{Cas:J18}%
  \BibitemOpen
  \bibfield  {author} {\bibinfo {author} {\bibfnamefont {D.}~\bibnamefont
  {Castelvecchi}},\ }\bibfield  {title} {\bibinfo {title} {The quantum internet
  has arrived (and it hasn’t)},\ }\href@noop {} {\bibfield  {journal}
  {\bibinfo  {journal} {Nature}\ }\textbf {\bibinfo {volume} {554}},\ \bibinfo
  {pages} {289} (\bibinfo {year} {2018})}\BibitemShut {NoStop}%
\bibitem [{\citenamefont {Pant}\ \emph {et~al.}(2019)\citenamefont {Pant},
  \citenamefont {Krovi}, \citenamefont {Towsley}, \citenamefont {Tassiulas},
  \citenamefont {Jiang}, \citenamefont {Basu}, \citenamefont {Englund},\ and\
  \citenamefont {Guha}}]{Panetal:J19}%
  \BibitemOpen
  \bibfield  {author} {\bibinfo {author} {\bibfnamefont {M.}~\bibnamefont
  {Pant}}, \bibinfo {author} {\bibfnamefont {H.}~\bibnamefont {Krovi}},
  \bibinfo {author} {\bibfnamefont {D.}~\bibnamefont {Towsley}}, \bibinfo
  {author} {\bibfnamefont {L.}~\bibnamefont {Tassiulas}}, \bibinfo {author}
  {\bibfnamefont {L.}~\bibnamefont {Jiang}}, \bibinfo {author} {\bibfnamefont
  {P.}~\bibnamefont {Basu}}, \bibinfo {author} {\bibfnamefont {D.}~\bibnamefont
  {Englund}},\ and\ \bibinfo {author} {\bibfnamefont {S.}~\bibnamefont
  {Guha}},\ }\bibfield  {title} {\bibinfo {title} {Routing entanglement in the
  quantum internet},\ }\href@noop {} {\bibfield  {journal} {\bibinfo  {journal}
  {npj {Q}uantum {I}nformation}\ }\textbf {\bibinfo {volume} {5}},\ \bibinfo
  {pages} {25} (\bibinfo {year} {2019})}\BibitemShut {NoStop}%
\bibitem [{\citenamefont {Arnon-Friedman}\ \emph {et~al.}(2018)\citenamefont
  {Arnon-Friedman}, \citenamefont {Dupuis}, \citenamefont {Fawzi},
  \citenamefont {Renner},\ and\ \citenamefont {Vidick}}]{ArnDupFawRenVid:J18}%
  \BibitemOpen
  \bibfield  {author} {\bibinfo {author} {\bibfnamefont {R.}~\bibnamefont
  {Arnon-Friedman}}, \bibinfo {author} {\bibfnamefont {F.}~\bibnamefont
  {Dupuis}}, \bibinfo {author} {\bibfnamefont {O.}~\bibnamefont {Fawzi}},
  \bibinfo {author} {\bibfnamefont {R.}~\bibnamefont {Renner}},\ and\ \bibinfo
  {author} {\bibfnamefont {T.}~\bibnamefont {Vidick}},\ }\bibfield  {title}
  {\bibinfo {title} {Practical device-independent quantum cryptography via
  entropy accumulation},\ }\href {https://doi.org/10.1038/s41467-017-02307-4}
  {\bibfield  {journal} {\bibinfo  {journal} {Nature Communications}\ }\textbf
  {\bibinfo {volume} {9}},\ \bibinfo {pages} {459} (\bibinfo {year}
  {2018})}\BibitemShut {NoStop}%
\bibitem [{\citenamefont {Avesani}\ \emph {et~al.}(2018)\citenamefont
  {Avesani}, \citenamefont {Marangon}, \citenamefont {Vallone},\ and\
  \citenamefont {Villoresi}}]{AveMarValVil:J18}%
  \BibitemOpen
  \bibfield  {author} {\bibinfo {author} {\bibfnamefont {M.}~\bibnamefont
  {Avesani}}, \bibinfo {author} {\bibfnamefont {D.~G.}\ \bibnamefont
  {Marangon}}, \bibinfo {author} {\bibfnamefont {G.}~\bibnamefont {Vallone}},\
  and\ \bibinfo {author} {\bibfnamefont {P.}~\bibnamefont {Villoresi}},\
  }\bibfield  {title} {\bibinfo {title} {Source-device-independent
  heterodyne-based quantum random number generator at 17 gbps},\ }\href
  {https://doi.org/10.1038/s41467-018-07585-0} {\bibfield  {journal} {\bibinfo
  {journal} {Nature Communications}\ }\textbf {\bibinfo {volume} {9}},\
  \bibinfo {pages} {5365} (\bibinfo {year} {2018})}\BibitemShut {NoStop}%
\bibitem [{\citenamefont {Lee}\ and\ \citenamefont {Hoban}(2018)}]{LeeHob:J18}%
  \BibitemOpen
  \bibfield  {author} {\bibinfo {author} {\bibfnamefont {C.~M.}\ \bibnamefont
  {Lee}}\ and\ \bibinfo {author} {\bibfnamefont {M.~J.}\ \bibnamefont
  {Hoban}},\ }\bibfield  {title} {\bibinfo {title} {Towards device-independent
  information processing on general quantum networks},\ }\href
  {https://doi.org/10.1103/PhysRevLett.120.020504} {\bibfield  {journal}
  {\bibinfo  {journal} {Phys. Rev. Lett.}\ }\textbf {\bibinfo {volume} {120}},\
  \bibinfo {pages} {020504} (\bibinfo {year} {2018})}\BibitemShut {NoStop}%
\bibitem [{\citenamefont {Schwonnek}\ \emph {et~al.}(2021)\citenamefont
  {Schwonnek}, \citenamefont {Goh}, \citenamefont {Primaatmaja}, \citenamefont
  {Tan}, \citenamefont {Wolf}, \citenamefont {Scarani},\ and\ \citenamefont
  {Lim}}]{Sch-etal:J21}%
  \BibitemOpen
  \bibfield  {author} {\bibinfo {author} {\bibfnamefont {R.}~\bibnamefont
  {Schwonnek}}, \bibinfo {author} {\bibfnamefont {K.~T.}\ \bibnamefont {Goh}},
  \bibinfo {author} {\bibfnamefont {I.~W.}\ \bibnamefont {Primaatmaja}},
  \bibinfo {author} {\bibfnamefont {E.~Y.~Z.}\ \bibnamefont {Tan}}, \bibinfo
  {author} {\bibfnamefont {R.}~\bibnamefont {Wolf}}, \bibinfo {author}
  {\bibfnamefont {V.}~\bibnamefont {Scarani}},\ and\ \bibinfo {author}
  {\bibfnamefont {C.~C.~W.}\ \bibnamefont {Lim}},\ }\bibfield  {title}
  {\bibinfo {title} {Device-independent quantum key distribution with random
  key basis},\ }\href {https://doi.org/10.1038/s41467-021-23147-3} {\bibfield
  {journal} {\bibinfo  {journal} {Nature Communications}\ }\textbf {\bibinfo
  {volume} {12}},\ \bibinfo {pages} {2880} (\bibinfo {year}
  {2021})}\BibitemShut {NoStop}%
\bibitem [{\citenamefont {Li}\ \emph {et~al.}(2021)\citenamefont {Li},
  \citenamefont {Zhang}, \citenamefont {Liu}, \citenamefont {Zhao},
  \citenamefont {Bai}, \citenamefont {Liu}, \citenamefont {Zhao}, \citenamefont
  {Peng}, \citenamefont {Zhang}, \citenamefont {Zhang}, \citenamefont {Munro},
  \citenamefont {Ma}, \citenamefont {Zhang}, \citenamefont {Fan},\ and\
  \citenamefont {Pan}}]{Lietal:J21}%
  \BibitemOpen
  \bibfield  {author} {\bibinfo {author} {\bibfnamefont {M.-H.}\ \bibnamefont
  {Li}}, \bibinfo {author} {\bibfnamefont {X.}~\bibnamefont {Zhang}}, \bibinfo
  {author} {\bibfnamefont {W.-Z.}\ \bibnamefont {Liu}}, \bibinfo {author}
  {\bibfnamefont {S.-R.}\ \bibnamefont {Zhao}}, \bibinfo {author}
  {\bibfnamefont {B.}~\bibnamefont {Bai}}, \bibinfo {author} {\bibfnamefont
  {Y.}~\bibnamefont {Liu}}, \bibinfo {author} {\bibfnamefont {Q.}~\bibnamefont
  {Zhao}}, \bibinfo {author} {\bibfnamefont {Y.}~\bibnamefont {Peng}}, \bibinfo
  {author} {\bibfnamefont {J.}~\bibnamefont {Zhang}}, \bibinfo {author}
  {\bibfnamefont {Y.}~\bibnamefont {Zhang}}, \bibinfo {author} {\bibfnamefont
  {W.~J.}\ \bibnamefont {Munro}}, \bibinfo {author} {\bibfnamefont
  {X.}~\bibnamefont {Ma}}, \bibinfo {author} {\bibfnamefont {Q.}~\bibnamefont
  {Zhang}}, \bibinfo {author} {\bibfnamefont {J.}~\bibnamefont {Fan}},\ and\
  \bibinfo {author} {\bibfnamefont {J.-W.}\ \bibnamefont {Pan}},\ }\bibfield
  {title} {\bibinfo {title} {Experimental realization of device-independent
  quantum randomness expansion},\ }\href
  {https://doi.org/10.1103/PhysRevLett.126.050503} {\bibfield  {journal}
  {\bibinfo  {journal} {Phys. Rev. Lett.}\ }\textbf {\bibinfo {volume} {126}},\
  \bibinfo {pages} {050503} (\bibinfo {year} {2021})}\BibitemShut {NoStop}%
\bibitem [{\citenamefont {Pirandola}\ \emph {et~al.}(2020)\citenamefont
  {Pirandola}, \citenamefont {Andersen}, \citenamefont {Banchi}, \citenamefont
  {Berta}, \citenamefont {Bunandar}, \citenamefont {Colbeck}, \citenamefont
  {Englund}, \citenamefont {Gehring}, \citenamefont {Lupo}, \citenamefont
  {Ottaviani}, \citenamefont {Pereira}, \citenamefont {Razavi}, \citenamefont
  {Shaari}, \citenamefont {Tomamichel}, \citenamefont {Usenko}, \citenamefont
  {Vallone}, \citenamefont {Villoresi},\ and\ \citenamefont
  {Wallden}}]{Pir-etal:J20}%
  \BibitemOpen
  \bibfield  {author} {\bibinfo {author} {\bibfnamefont {S.}~\bibnamefont
  {Pirandola}}, \bibinfo {author} {\bibfnamefont {U.~L.}\ \bibnamefont
  {Andersen}}, \bibinfo {author} {\bibfnamefont {L.}~\bibnamefont {Banchi}},
  \bibinfo {author} {\bibfnamefont {M.}~\bibnamefont {Berta}}, \bibinfo
  {author} {\bibfnamefont {D.}~\bibnamefont {Bunandar}}, \bibinfo {author}
  {\bibfnamefont {R.}~\bibnamefont {Colbeck}}, \bibinfo {author} {\bibfnamefont
  {D.}~\bibnamefont {Englund}}, \bibinfo {author} {\bibfnamefont
  {T.}~\bibnamefont {Gehring}}, \bibinfo {author} {\bibfnamefont
  {C.}~\bibnamefont {Lupo}}, \bibinfo {author} {\bibfnamefont {C.}~\bibnamefont
  {Ottaviani}}, \bibinfo {author} {\bibfnamefont {J.~L.}\ \bibnamefont
  {Pereira}}, \bibinfo {author} {\bibfnamefont {M.}~\bibnamefont {Razavi}},
  \bibinfo {author} {\bibfnamefont {J.~S.}\ \bibnamefont {Shaari}}, \bibinfo
  {author} {\bibfnamefont {M.}~\bibnamefont {Tomamichel}}, \bibinfo {author}
  {\bibfnamefont {V.~C.}\ \bibnamefont {Usenko}}, \bibinfo {author}
  {\bibfnamefont {G.}~\bibnamefont {Vallone}}, \bibinfo {author} {\bibfnamefont
  {P.}~\bibnamefont {Villoresi}},\ and\ \bibinfo {author} {\bibfnamefont
  {P.}~\bibnamefont {Wallden}},\ }\bibfield  {title} {\bibinfo {title}
  {Advances in quantum cryptography},\ }\href
  {https://doi.org/10.1364/AOP.361502} {\bibfield  {journal} {\bibinfo
  {journal} {Adv. Opt. Photon.}\ }\textbf {\bibinfo {volume} {12}},\ \bibinfo
  {pages} {1012} (\bibinfo {year} {2020})}\BibitemShut {NoStop}%
\bibitem [{\citenamefont {K{\"o}rber}\ \emph {et~al.}(2018)\citenamefont
  {K{\"o}rber}, \citenamefont {Morin}, \citenamefont {Langenfeld},
  \citenamefont {Neuzner}, \citenamefont {Ritter},\ and\ \citenamefont
  {Rempe}}]{KorMorLanNeuRitRem:J18}%
  \BibitemOpen
  \bibfield  {author} {\bibinfo {author} {\bibfnamefont {M.}~\bibnamefont
  {K{\"o}rber}}, \bibinfo {author} {\bibfnamefont {O.}~\bibnamefont {Morin}},
  \bibinfo {author} {\bibfnamefont {S.}~\bibnamefont {Langenfeld}}, \bibinfo
  {author} {\bibfnamefont {A.}~\bibnamefont {Neuzner}}, \bibinfo {author}
  {\bibfnamefont {S.}~\bibnamefont {Ritter}},\ and\ \bibinfo {author}
  {\bibfnamefont {G.}~\bibnamefont {Rempe}},\ }\bibfield  {title} {\bibinfo
  {title} {Decoherence-protected memory for a single-photon qubit},\ }\href
  {https://doi.org/10.1038/s41566-017-0050-y} {\bibfield  {journal} {\bibinfo
  {journal} {Nature Photonics}\ }\textbf {\bibinfo {volume} {12}},\ \bibinfo
  {pages} {18} (\bibinfo {year} {2018})}\BibitemShut {NoStop}%
\bibitem [{\citenamefont {Yin}\ \emph {et~al.}(2017)\citenamefont {Yin},
  \citenamefont {Cao}, \citenamefont {Li}, \citenamefont {Liao}, \citenamefont
  {Zhang}, \citenamefont {Ren}, \citenamefont {Cai}, \citenamefont {Liu},
  \citenamefont {Li}, \citenamefont {Dai}, \citenamefont {Li}, \citenamefont
  {Lu}, \citenamefont {Gong}, \citenamefont {Xu}, \citenamefont {Li},
  \citenamefont {Li}, \citenamefont {Yin}, \citenamefont {Jiang}, \citenamefont
  {Li}, \citenamefont {Jia}, \citenamefont {Ren}, \citenamefont {He},
  \citenamefont {Zhou}, \citenamefont {Zhang}, \citenamefont {Wang},
  \citenamefont {Chang}, \citenamefont {Zhu}, \citenamefont {Liu},
  \citenamefont {Chen}, \citenamefont {Lu}, \citenamefont {Shu}, \citenamefont
  {Peng}, \citenamefont {Wang},\ and\ \citenamefont {Pan}}]{YinCaoPanetal:J17}%
  \BibitemOpen
  \bibfield  {author} {\bibinfo {author} {\bibfnamefont {J.}~\bibnamefont
  {Yin}}, \bibinfo {author} {\bibfnamefont {Y.}~\bibnamefont {Cao}}, \bibinfo
  {author} {\bibfnamefont {Y.-H.}\ \bibnamefont {Li}}, \bibinfo {author}
  {\bibfnamefont {S.-K.}\ \bibnamefont {Liao}}, \bibinfo {author}
  {\bibfnamefont {L.}~\bibnamefont {Zhang}}, \bibinfo {author} {\bibfnamefont
  {J.-G.}\ \bibnamefont {Ren}}, \bibinfo {author} {\bibfnamefont {W.-Q.}\
  \bibnamefont {Cai}}, \bibinfo {author} {\bibfnamefont {W.-Y.}\ \bibnamefont
  {Liu}}, \bibinfo {author} {\bibfnamefont {B.}~\bibnamefont {Li}}, \bibinfo
  {author} {\bibfnamefont {H.}~\bibnamefont {Dai}}, \bibinfo {author}
  {\bibfnamefont {G.-B.}\ \bibnamefont {Li}}, \bibinfo {author} {\bibfnamefont
  {Q.-M.}\ \bibnamefont {Lu}}, \bibinfo {author} {\bibfnamefont {Y.-H.}\
  \bibnamefont {Gong}}, \bibinfo {author} {\bibfnamefont {Y.}~\bibnamefont
  {Xu}}, \bibinfo {author} {\bibfnamefont {S.-L.}\ \bibnamefont {Li}}, \bibinfo
  {author} {\bibfnamefont {F.-Z.}\ \bibnamefont {Li}}, \bibinfo {author}
  {\bibfnamefont {Y.-Y.}\ \bibnamefont {Yin}}, \bibinfo {author} {\bibfnamefont
  {Z.-Q.}\ \bibnamefont {Jiang}}, \bibinfo {author} {\bibfnamefont
  {M.}~\bibnamefont {Li}}, \bibinfo {author} {\bibfnamefont {J.-J.}\
  \bibnamefont {Jia}}, \bibinfo {author} {\bibfnamefont {G.}~\bibnamefont
  {Ren}}, \bibinfo {author} {\bibfnamefont {D.}~\bibnamefont {He}}, \bibinfo
  {author} {\bibfnamefont {Y.-L.}\ \bibnamefont {Zhou}}, \bibinfo {author}
  {\bibfnamefont {X.-X.}\ \bibnamefont {Zhang}}, \bibinfo {author}
  {\bibfnamefont {N.}~\bibnamefont {Wang}}, \bibinfo {author} {\bibfnamefont
  {X.}~\bibnamefont {Chang}}, \bibinfo {author} {\bibfnamefont {Z.-C.}\
  \bibnamefont {Zhu}}, \bibinfo {author} {\bibfnamefont {N.-L.}\ \bibnamefont
  {Liu}}, \bibinfo {author} {\bibfnamefont {Y.-A.}\ \bibnamefont {Chen}},
  \bibinfo {author} {\bibfnamefont {C.-Y.}\ \bibnamefont {Lu}}, \bibinfo
  {author} {\bibfnamefont {R.}~\bibnamefont {Shu}}, \bibinfo {author}
  {\bibfnamefont {C.-Z.}\ \bibnamefont {Peng}}, \bibinfo {author}
  {\bibfnamefont {J.-Y.}\ \bibnamefont {Wang}},\ and\ \bibinfo {author}
  {\bibfnamefont {J.-W.}\ \bibnamefont {Pan}},\ }\bibfield  {title} {\bibinfo
  {title} {Satellite-based entanglement distribution over 1200 kilometers},\
  }\href {https://doi.org/10.1126/science.aan3211} {\bibfield  {journal}
  {\bibinfo  {journal} {Science}\ }\textbf {\bibinfo {volume} {356}},\ \bibinfo
  {pages} {1140} (\bibinfo {year} {2017})}\BibitemShut {NoStop}%
\bibitem [{\citenamefont {Humphreys}\ \emph {et~al.}(2018)\citenamefont
  {Humphreys}, \citenamefont {Kalb}, \citenamefont {Morits}, \citenamefont
  {Schouten}, \citenamefont {Vermeulen}, \citenamefont {Twitchen},
  \citenamefont {Markham},\ and\ \citenamefont {Hanson}}]{Humetal:J18}%
  \BibitemOpen
  \bibfield  {author} {\bibinfo {author} {\bibfnamefont {P.~C.}\ \bibnamefont
  {Humphreys}}, \bibinfo {author} {\bibfnamefont {N.}~\bibnamefont {Kalb}},
  \bibinfo {author} {\bibfnamefont {J.~P.~J.}\ \bibnamefont {Morits}}, \bibinfo
  {author} {\bibfnamefont {R.~N.}\ \bibnamefont {Schouten}}, \bibinfo {author}
  {\bibfnamefont {R.~F.~L.}\ \bibnamefont {Vermeulen}}, \bibinfo {author}
  {\bibfnamefont {D.~J.}\ \bibnamefont {Twitchen}}, \bibinfo {author}
  {\bibfnamefont {M.}~\bibnamefont {Markham}},\ and\ \bibinfo {author}
  {\bibfnamefont {R.}~\bibnamefont {Hanson}},\ }\bibfield  {title} {\bibinfo
  {title} {Deterministic delivery of remote entanglement on a quantum
  network},\ }\href@noop {} {\bibfield  {journal} {\bibinfo  {journal}
  {Nature}\ }\textbf {\bibinfo {volume} {558}},\ \bibinfo {pages} {268}
  (\bibinfo {year} {2018})}\BibitemShut {NoStop}%
\bibitem [{\citenamefont {Pompili}\ \emph {et~al.}(2021)\citenamefont
  {Pompili}, \citenamefont {Hermans}, \citenamefont {Baier}, \citenamefont
  {Beukers}, \citenamefont {Humphreys}, \citenamefont {Schouten}, \citenamefont
  {Vermeulen}, \citenamefont {Tiggelman}, \citenamefont {dos Santos~Martins},
  \citenamefont {Dirkse}, \citenamefont {Wehner},\ and\ \citenamefont
  {Hanson}}]{Pom-etal:J21}%
  \BibitemOpen
  \bibfield  {author} {\bibinfo {author} {\bibfnamefont {M.}~\bibnamefont
  {Pompili}}, \bibinfo {author} {\bibfnamefont {S.~L.~N.}\ \bibnamefont
  {Hermans}}, \bibinfo {author} {\bibfnamefont {S.}~\bibnamefont {Baier}},
  \bibinfo {author} {\bibfnamefont {H.~K.~C.}\ \bibnamefont {Beukers}},
  \bibinfo {author} {\bibfnamefont {P.~C.}\ \bibnamefont {Humphreys}}, \bibinfo
  {author} {\bibfnamefont {R.~N.}\ \bibnamefont {Schouten}}, \bibinfo {author}
  {\bibfnamefont {R.~F.~L.}\ \bibnamefont {Vermeulen}}, \bibinfo {author}
  {\bibfnamefont {M.~J.}\ \bibnamefont {Tiggelman}}, \bibinfo {author}
  {\bibfnamefont {L.}~\bibnamefont {dos Santos~Martins}}, \bibinfo {author}
  {\bibfnamefont {B.}~\bibnamefont {Dirkse}}, \bibinfo {author} {\bibfnamefont
  {S.}~\bibnamefont {Wehner}},\ and\ \bibinfo {author} {\bibfnamefont
  {R.}~\bibnamefont {Hanson}},\ }\bibfield  {title} {\bibinfo {title}
  {Realization of a multinode quantum network of remote solid-state qubits},\
  }\href {https://doi.org/10.1126/science.abg1919} {\bibfield  {journal}
  {\bibinfo  {journal} {Science}\ }\textbf {\bibinfo {volume} {372}},\ \bibinfo
  {pages} {259} (\bibinfo {year} {2021})}\BibitemShut {NoStop}%
\bibitem [{\citenamefont {Liu}\ \emph {et~al.}(2020)\citenamefont {Liu},
  \citenamefont {Tian}, \citenamefont {Gu}, \citenamefont {Fan}, \citenamefont
  {Ni}, \citenamefont {Yang}, \citenamefont {Zhang}, \citenamefont {Hu},
  \citenamefont {Guo}, \citenamefont {Cao}, \citenamefont {Hu}, \citenamefont
  {Zhao}, \citenamefont {Lu}, \citenamefont {Gong}, \citenamefont {Xie},\ and\
  \citenamefont {Zhu}}]{LiuTianGuetal:J20}%
  \BibitemOpen
  \bibfield  {author} {\bibinfo {author} {\bibfnamefont {H.-Y.}\ \bibnamefont
  {Liu}}, \bibinfo {author} {\bibfnamefont {X.-H.}\ \bibnamefont {Tian}},
  \bibinfo {author} {\bibfnamefont {C.}~\bibnamefont {Gu}}, \bibinfo {author}
  {\bibfnamefont {P.}~\bibnamefont {Fan}}, \bibinfo {author} {\bibfnamefont
  {X.}~\bibnamefont {Ni}}, \bibinfo {author} {\bibfnamefont {R.}~\bibnamefont
  {Yang}}, \bibinfo {author} {\bibfnamefont {J.-N.}\ \bibnamefont {Zhang}},
  \bibinfo {author} {\bibfnamefont {M.}~\bibnamefont {Hu}}, \bibinfo {author}
  {\bibfnamefont {J.}~\bibnamefont {Guo}}, \bibinfo {author} {\bibfnamefont
  {X.}~\bibnamefont {Cao}}, \bibinfo {author} {\bibfnamefont {X.}~\bibnamefont
  {Hu}}, \bibinfo {author} {\bibfnamefont {G.}~\bibnamefont {Zhao}}, \bibinfo
  {author} {\bibfnamefont {Y.-Q.}\ \bibnamefont {Lu}}, \bibinfo {author}
  {\bibfnamefont {Y.-X.}\ \bibnamefont {Gong}}, \bibinfo {author}
  {\bibfnamefont {Z.}~\bibnamefont {Xie}},\ and\ \bibinfo {author}
  {\bibfnamefont {S.-N.}\ \bibnamefont {Zhu}},\ }\bibfield  {title} {\bibinfo
  {title} {{Drone-based entanglement distribution towards mobile quantum
  networks}},\ }\href@noop {} {\bibfield  {journal} {\bibinfo  {journal}
  {National Science Review}\ }\textbf {\bibinfo {volume} {7}},\ \bibinfo
  {pages} {921} (\bibinfo {year} {2020})}\BibitemShut {NoStop}%
\bibitem [{\citenamefont {West}\ \emph {et~al.}(2010)\citenamefont {West},
  \citenamefont {Lidar}, \citenamefont {Fong},\ and\ \citenamefont
  {Gyure}}]{WesLidFonGyu:J10}%
  \BibitemOpen
  \bibfield  {author} {\bibinfo {author} {\bibfnamefont {J.~R.}\ \bibnamefont
  {West}}, \bibinfo {author} {\bibfnamefont {D.~A.}\ \bibnamefont {Lidar}},
  \bibinfo {author} {\bibfnamefont {B.~H.}\ \bibnamefont {Fong}},\ and\
  \bibinfo {author} {\bibfnamefont {M.~F.}\ \bibnamefont {Gyure}},\ }\bibfield
  {title} {\bibinfo {title} {High fidelity quantum gates via dynamical
  decoupling},\ }\href@noop {} {\bibfield  {journal} {\bibinfo  {journal}
  {Phys. Rev. Lett.}\ }\textbf {\bibinfo {volume} {105}},\ \bibinfo {pages}
  {230503} (\bibinfo {year} {2010})}\BibitemShut {NoStop}%
\bibitem [{\citenamefont {Bason}\ \emph {et~al.}(2011)\citenamefont {Bason},
  \citenamefont {Viteau}, \citenamefont {Malossi}, \citenamefont {Huillery},
  \citenamefont {Arimondo}, \citenamefont {Ciampini}, \citenamefont {Fazio},
  \citenamefont {Giovannetti}, \citenamefont {Mannella},\ and\ \citenamefont
  {Morsch}}]{Basetal:J11}%
  \BibitemOpen
  \bibfield  {author} {\bibinfo {author} {\bibfnamefont {M.~G.}\ \bibnamefont
  {Bason}}, \bibinfo {author} {\bibfnamefont {M.}~\bibnamefont {Viteau}},
  \bibinfo {author} {\bibfnamefont {N.}~\bibnamefont {Malossi}}, \bibinfo
  {author} {\bibfnamefont {P.}~\bibnamefont {Huillery}}, \bibinfo {author}
  {\bibfnamefont {E.}~\bibnamefont {Arimondo}}, \bibinfo {author}
  {\bibfnamefont {D.}~\bibnamefont {Ciampini}}, \bibinfo {author}
  {\bibfnamefont {R.}~\bibnamefont {Fazio}}, \bibinfo {author} {\bibfnamefont
  {V.}~\bibnamefont {Giovannetti}}, \bibinfo {author} {\bibfnamefont
  {R.}~\bibnamefont {Mannella}},\ and\ \bibinfo {author} {\bibfnamefont
  {O.}~\bibnamefont {Morsch}},\ }\bibfield  {title} {\bibinfo {title}
  {High-fidelity quantum driving},\ }\href@noop {} {\bibfield  {journal}
  {\bibinfo  {journal} {Nature Physics}\ }\textbf {\bibinfo {volume} {8}},\
  \bibinfo {pages} {147} (\bibinfo {year} {2011})}\BibitemShut {NoStop}%
\bibitem [{\citenamefont {Arroyo-Camejo}\ \emph {et~al.}(2014)\citenamefont
  {Arroyo-Camejo}, \citenamefont {Lazariev}, \citenamefont {Hell},\ and\
  \citenamefont {Balasubramanian}}]{ArrLazHelBal:J14}%
  \BibitemOpen
  \bibfield  {author} {\bibinfo {author} {\bibfnamefont {S.}~\bibnamefont
  {Arroyo-Camejo}}, \bibinfo {author} {\bibfnamefont {A.}~\bibnamefont
  {Lazariev}}, \bibinfo {author} {\bibfnamefont {S.~W.}\ \bibnamefont {Hell}},\
  and\ \bibinfo {author} {\bibfnamefont {G.}~\bibnamefont {Balasubramanian}},\
  }\bibfield  {title} {\bibinfo {title} {Room temperature high-fidelity
  holonomic single-qubit gate on a solid-state spin},\ }\href@noop {}
  {\bibfield  {journal} {\bibinfo  {journal} {Nature Communications}\ }\textbf
  {\bibinfo {volume} {5}},\ \bibinfo {pages} {4870 EP} (\bibinfo {year}
  {2014})}\BibitemShut {NoStop}%
\bibitem [{\citenamefont {Zhang}\ \emph {et~al.}(2021)\citenamefont {Zhang},
  \citenamefont {Luo}, \citenamefont {Wen}, \citenamefont {Feng}, \citenamefont
  {Pang}, \citenamefont {Luo},\ and\ \citenamefont
  {Zhou}}]{ZhaLuoWen-etal:J21}%
  \BibitemOpen
  \bibfield  {author} {\bibinfo {author} {\bibfnamefont {X.}~\bibnamefont
  {Zhang}}, \bibinfo {author} {\bibfnamefont {M.}~\bibnamefont {Luo}}, \bibinfo
  {author} {\bibfnamefont {Z.}~\bibnamefont {Wen}}, \bibinfo {author}
  {\bibfnamefont {Q.}~\bibnamefont {Feng}}, \bibinfo {author} {\bibfnamefont
  {S.}~\bibnamefont {Pang}}, \bibinfo {author} {\bibfnamefont {W.}~\bibnamefont
  {Luo}},\ and\ \bibinfo {author} {\bibfnamefont {X.}~\bibnamefont {Zhou}},\
  }\bibfield  {title} {\bibinfo {title} {Direct fidelity estimation of quantum
  states using machine learning},\ }\href
  {https://doi.org/10.1103/PhysRevLett.127.130503} {\bibfield  {journal}
  {\bibinfo  {journal} {Phys. Rev. Lett.}\ }\textbf {\bibinfo {volume} {127}},\
  \bibinfo {pages} {130503} (\bibinfo {year} {2021})}\BibitemShut {NoStop}%
\bibitem [{\citenamefont {Somma}\ \emph {et~al.}(2006)\citenamefont {Somma},
  \citenamefont {Chiaverini},\ and\ \citenamefont {Berkeland}}]{SomChiBer:J06}%
  \BibitemOpen
  \bibfield  {author} {\bibinfo {author} {\bibfnamefont {R.~D.}\ \bibnamefont
  {Somma}}, \bibinfo {author} {\bibfnamefont {J.}~\bibnamefont {Chiaverini}},\
  and\ \bibinfo {author} {\bibfnamefont {D.~J.}\ \bibnamefont {Berkeland}},\
  }\bibfield  {title} {\bibinfo {title} {Lower bounds for the fidelity of
  entangled-state preparation},\ }\href@noop {} {\bibfield  {journal} {\bibinfo
   {journal} {Phys. Rev. A}\ }\textbf {\bibinfo {volume} {74}},\ \bibinfo
  {pages} {052302} (\bibinfo {year} {2006})}\BibitemShut {NoStop}%
\bibitem [{\citenamefont {G\"uhne}\ \emph {et~al.}(2007)\citenamefont
  {G\"uhne}, \citenamefont {Lu}, \citenamefont {Gao},\ and\ \citenamefont
  {Pan}}]{GuhLuGaoPan:J07}%
  \BibitemOpen
  \bibfield  {author} {\bibinfo {author} {\bibfnamefont {O.}~\bibnamefont
  {G\"uhne}}, \bibinfo {author} {\bibfnamefont {C.-Y.}\ \bibnamefont {Lu}},
  \bibinfo {author} {\bibfnamefont {W.-B.}\ \bibnamefont {Gao}},\ and\ \bibinfo
  {author} {\bibfnamefont {J.-W.}\ \bibnamefont {Pan}},\ }\bibfield  {title}
  {\bibinfo {title} {Toolbox for entanglement detection and fidelity
  estimation},\ }\href@noop {} {\bibfield  {journal} {\bibinfo  {journal}
  {Phys. Rev. A}\ }\textbf {\bibinfo {volume} {76}},\ \bibinfo {pages} {030305}
  (\bibinfo {year} {2007})}\BibitemShut {NoStop}%
\bibitem [{\citenamefont {Flammia}\ and\ \citenamefont
  {Liu}(2011)}]{FlaLiu:J11}%
  \BibitemOpen
  \bibfield  {author} {\bibinfo {author} {\bibfnamefont {S.~T.}\ \bibnamefont
  {Flammia}}\ and\ \bibinfo {author} {\bibfnamefont {Y.-K.}\ \bibnamefont
  {Liu}},\ }\bibfield  {title} {\bibinfo {title} {Direct fidelity estimation
  from few pauli measurements},\ }\href@noop {} {\bibfield  {journal} {\bibinfo
   {journal} {Phys. Rev. Lett.}\ }\textbf {\bibinfo {volume} {106}},\ \bibinfo
  {pages} {230501} (\bibinfo {year} {2011})}\BibitemShut {NoStop}%
\bibitem [{\citenamefont {Zhu}\ and\ \citenamefont
  {Hayashi}(2019)}]{ZhuHay:J19}%
  \BibitemOpen
  \bibfield  {author} {\bibinfo {author} {\bibfnamefont {H.}~\bibnamefont
  {Zhu}}\ and\ \bibinfo {author} {\bibfnamefont {M.}~\bibnamefont {Hayashi}},\
  }\bibfield  {title} {\bibinfo {title} {Optimal verification and fidelity
  estimation of maximally entangled states},\ }\href@noop {} {\bibfield
  {journal} {\bibinfo  {journal} {Phys. Rev. A}\ }\textbf {\bibinfo {volume}
  {99}},\ \bibinfo {pages} {052346} (\bibinfo {year} {2019})}\BibitemShut
  {NoStop}%
\bibitem [{\citenamefont {Kalb}\ \emph {et~al.}(2017)\citenamefont {Kalb},
  \citenamefont {Reiserer}, \citenamefont {Humphreys}, \citenamefont
  {Bakermans}, \citenamefont {Kamerling}, \citenamefont {Nickerson},
  \citenamefont {Benjamin}, \citenamefont {Twitchen}, \citenamefont {Markham},\
  and\ \citenamefont {Hanson}}]{Kaletal:J17}%
  \BibitemOpen
  \bibfield  {author} {\bibinfo {author} {\bibfnamefont {N.}~\bibnamefont
  {Kalb}}, \bibinfo {author} {\bibfnamefont {A.~A.}\ \bibnamefont {Reiserer}},
  \bibinfo {author} {\bibfnamefont {P.~C.}\ \bibnamefont {Humphreys}}, \bibinfo
  {author} {\bibfnamefont {J.~J.~W.}\ \bibnamefont {Bakermans}}, \bibinfo
  {author} {\bibfnamefont {S.~J.}\ \bibnamefont {Kamerling}}, \bibinfo {author}
  {\bibfnamefont {N.~H.}\ \bibnamefont {Nickerson}}, \bibinfo {author}
  {\bibfnamefont {S.~C.}\ \bibnamefont {Benjamin}}, \bibinfo {author}
  {\bibfnamefont {D.~J.}\ \bibnamefont {Twitchen}}, \bibinfo {author}
  {\bibfnamefont {M.}~\bibnamefont {Markham}},\ and\ \bibinfo {author}
  {\bibfnamefont {R.}~\bibnamefont {Hanson}},\ }\bibfield  {title} {\bibinfo
  {title} {Entanglement distillation between solid-state quantum network
  nodes},\ }\href@noop {} {\bibfield  {journal} {\bibinfo  {journal} {Science}\
  }\textbf {\bibinfo {volume} {356}},\ \bibinfo {pages} {928} (\bibinfo {year}
  {2017})}\BibitemShut {NoStop}%
\bibitem [{\citenamefont {Tomamichel}\ \emph {et~al.}(2012)\citenamefont
  {Tomamichel}, \citenamefont {Lim}, \citenamefont {Gisin},\ and\ \citenamefont
  {Renner}}]{TomLimGisRen:J12}%
  \BibitemOpen
  \bibfield  {author} {\bibinfo {author} {\bibfnamefont {M.}~\bibnamefont
  {Tomamichel}}, \bibinfo {author} {\bibfnamefont {C.~C.~W.}\ \bibnamefont
  {Lim}}, \bibinfo {author} {\bibfnamefont {N.}~\bibnamefont {Gisin}},\ and\
  \bibinfo {author} {\bibfnamefont {R.}~\bibnamefont {Renner}},\ }\bibfield
  {title} {\bibinfo {title} {Tight finite-key analysis for quantum
  cryptography},\ }\href@noop {} {\bibfield  {journal} {\bibinfo  {journal}
  {Nature Communications}\ }\textbf {\bibinfo {volume} {3}},\ \bibinfo {pages}
  {634 } (\bibinfo {year} {2012})}\BibitemShut {NoStop}%
\bibitem [{\citenamefont {Pfister}\ \emph {et~al.}(2018)\citenamefont
  {Pfister}, \citenamefont {Rol}, \citenamefont {Mantri}, \citenamefont
  {Tomamichel},\ and\ \citenamefont {Wehner}}]{PfiRolManTomWeh:J18}%
  \BibitemOpen
  \bibfield  {author} {\bibinfo {author} {\bibfnamefont {C.}~\bibnamefont
  {Pfister}}, \bibinfo {author} {\bibfnamefont {M.~A.}\ \bibnamefont {Rol}},
  \bibinfo {author} {\bibfnamefont {A.}~\bibnamefont {Mantri}}, \bibinfo
  {author} {\bibfnamefont {M.}~\bibnamefont {Tomamichel}},\ and\ \bibinfo
  {author} {\bibfnamefont {S.}~\bibnamefont {Wehner}},\ }\bibfield  {title}
  {\bibinfo {title} {Capacity estimation and verification of quantum channels
  with arbitrarily correlated errors},\ }\href@noop {} {\bibfield  {journal}
  {\bibinfo  {journal} {Nature Communications}\ }\textbf {\bibinfo {volume}
  {9}},\ \bibinfo {pages} {27} (\bibinfo {year} {2018})}\BibitemShut {NoStop}%
\bibitem [{\citenamefont {Bagan}\ \emph {et~al.}(2006)\citenamefont {Bagan},
  \citenamefont {Ballester}, \citenamefont {Gill}, \citenamefont {Mu\~noz
  Tapia},\ and\ \citenamefont {Romero-Isart}}]{BagBalGilRom:J06}%
  \BibitemOpen
  \bibfield  {author} {\bibinfo {author} {\bibfnamefont {E.}~\bibnamefont
  {Bagan}}, \bibinfo {author} {\bibfnamefont {M.~A.}\ \bibnamefont
  {Ballester}}, \bibinfo {author} {\bibfnamefont {R.~D.}\ \bibnamefont {Gill}},
  \bibinfo {author} {\bibfnamefont {R.}~\bibnamefont {Mu\~noz Tapia}},\ and\
  \bibinfo {author} {\bibfnamefont {O.}~\bibnamefont {Romero-Isart}},\
  }\bibfield  {title} {\bibinfo {title} {Separable measurement estimation of
  density matrices and its fidelity gap with collective protocols},\
  }\href@noop {} {\bibfield  {journal} {\bibinfo  {journal} {Phys. Rev. Lett.}\
  }\textbf {\bibinfo {volume} {97}},\ \bibinfo {pages} {130501} (\bibinfo
  {year} {2006})}\BibitemShut {NoStop}%
\bibitem [{\citenamefont {De~Greve}\ \emph {et~al.}(2013)\citenamefont
  {De~Greve}, \citenamefont {McMahon}, \citenamefont {Yu}, \citenamefont
  {Pelc}, \citenamefont {Jones}, \citenamefont {Natarajan}, \citenamefont
  {Kim}, \citenamefont {Abe}, \citenamefont {Maier}, \citenamefont {Schneider},
  \citenamefont {Kamp}, \citenamefont {H{\"o}fling}, \citenamefont {Hadfield},
  \citenamefont {Forchel}, \citenamefont {Fejer},\ and\ \citenamefont
  {Yamamoto}}]{DeGetal:J13}%
  \BibitemOpen
  \bibfield  {author} {\bibinfo {author} {\bibfnamefont {K.}~\bibnamefont
  {De~Greve}}, \bibinfo {author} {\bibfnamefont {P.~L.}\ \bibnamefont
  {McMahon}}, \bibinfo {author} {\bibfnamefont {L.}~\bibnamefont {Yu}},
  \bibinfo {author} {\bibfnamefont {J.~S.}\ \bibnamefont {Pelc}}, \bibinfo
  {author} {\bibfnamefont {C.}~\bibnamefont {Jones}}, \bibinfo {author}
  {\bibfnamefont {C.~M.}\ \bibnamefont {Natarajan}}, \bibinfo {author}
  {\bibfnamefont {N.~Y.}\ \bibnamefont {Kim}}, \bibinfo {author} {\bibfnamefont
  {E.}~\bibnamefont {Abe}}, \bibinfo {author} {\bibfnamefont {S.}~\bibnamefont
  {Maier}}, \bibinfo {author} {\bibfnamefont {C.}~\bibnamefont {Schneider}},
  \bibinfo {author} {\bibfnamefont {M.}~\bibnamefont {Kamp}}, \bibinfo {author}
  {\bibfnamefont {S.}~\bibnamefont {H{\"o}fling}}, \bibinfo {author}
  {\bibfnamefont {R.~H.}\ \bibnamefont {Hadfield}}, \bibinfo {author}
  {\bibfnamefont {A.}~\bibnamefont {Forchel}}, \bibinfo {author} {\bibfnamefont
  {M.~M.}\ \bibnamefont {Fejer}},\ and\ \bibinfo {author} {\bibfnamefont
  {Y.}~\bibnamefont {Yamamoto}},\ }\bibfield  {title} {\bibinfo {title}
  {Complete tomography of a high-fidelity solid-state entangled spin--photon
  qubit pair},\ }\href@noop {} {\bibfield  {journal} {\bibinfo  {journal}
  {Nature Communications}\ }\textbf {\bibinfo {volume} {4}},\ \bibinfo {pages}
  {2228 EP } (\bibinfo {year} {2013})}\BibitemShut {NoStop}%
\bibitem [{\citenamefont {Bock}\ \emph {et~al.}(2018)\citenamefont {Bock},
  \citenamefont {Eich}, \citenamefont {Kucera}, \citenamefont {Kreis},
  \citenamefont {Lenhard}, \citenamefont {Becher},\ and\ \citenamefont
  {Eschner}}]{Bocetal:J18}%
  \BibitemOpen
  \bibfield  {author} {\bibinfo {author} {\bibfnamefont {M.}~\bibnamefont
  {Bock}}, \bibinfo {author} {\bibfnamefont {P.}~\bibnamefont {Eich}}, \bibinfo
  {author} {\bibfnamefont {S.}~\bibnamefont {Kucera}}, \bibinfo {author}
  {\bibfnamefont {M.}~\bibnamefont {Kreis}}, \bibinfo {author} {\bibfnamefont
  {A.}~\bibnamefont {Lenhard}}, \bibinfo {author} {\bibfnamefont
  {C.}~\bibnamefont {Becher}},\ and\ \bibinfo {author} {\bibfnamefont
  {J.}~\bibnamefont {Eschner}},\ }\bibfield  {title} {\bibinfo {title}
  {High-fidelity entanglement between a trapped ion and a telecom photon via
  quantum frequency conversion},\ }\href@noop {} {\bibfield  {journal}
  {\bibinfo  {journal} {Nature Communications}\ }\textbf {\bibinfo {volume}
  {9}},\ \bibinfo {pages} {1998} (\bibinfo {year} {2018})}\BibitemShut
  {NoStop}%
\bibitem [{\citenamefont {Cram\'er}(1999)}]{Cra:B99}%
  \BibitemOpen
  \bibfield  {author} {\bibinfo {author} {\bibfnamefont {H.}~\bibnamefont
  {Cram\'er}},\ }\href@noop {} {\emph {\bibinfo {title} {Mathematical Methods
  of Statistics}}}\ (\bibinfo  {publisher} {Princeton University Press},\
  \bibinfo {address} {Princeton, New Jersey, U.S.},\ \bibinfo {year}
  {1999})\BibitemShut {NoStop}%
\bibitem [{\citenamefont {Rao}(1994)}]{Rao:B94}%
  \BibitemOpen
  \bibfield  {author} {\bibinfo {author} {\bibfnamefont {C.~R.}\ \bibnamefont
  {Rao}},\ }\href@noop {} {\emph {\bibinfo {title} {Selected Papers of C.R.
  Rao}}}\ (\bibinfo  {publisher} {John Wiley \& Sons},\ \bibinfo {address} {New
  York, U.S.},\ \bibinfo {year} {1994})\BibitemShut {NoStop}%
\bibitem [{\citenamefont {Paris}(2009)}]{Paris:09}%
  \BibitemOpen
  \bibfield  {author} {\bibinfo {author} {\bibfnamefont {M.~G.~A.}\
  \bibnamefont {Paris}},\ }\bibfield  {title} {\bibinfo {title} {Quantum
  estimation for quantum technology},\ }\href@noop {} {\bibfield  {journal}
  {\bibinfo  {journal} {International Journal of Quantum Information}\ }\textbf
  {\bibinfo {volume} {07}},\ \bibinfo {pages} {125} (\bibinfo {year}
  {2009})}\BibitemShut {NoStop}%
\bibitem [{\citenamefont {\ifmmode~\check{S}\else
  \v{S}\fi{}afr\'anek}(2018)}]{Saf:18}%
  \BibitemOpen
  \bibfield  {author} {\bibinfo {author} {\bibfnamefont {D.}~\bibnamefont
  {\ifmmode~\check{S}\else \v{S}\fi{}afr\'anek}},\ }\bibfield  {title}
  {\bibinfo {title} {Simple expression for the quantum fisher information
  matrix},\ }\href@noop {} {\bibfield  {journal} {\bibinfo  {journal} {Phys.
  Rev. A}\ }\textbf {\bibinfo {volume} {97}},\ \bibinfo {pages} {042322}
  (\bibinfo {year} {2018})}\BibitemShut {NoStop}%
\bibitem [{\citenamefont {Bennett}\ \emph {et~al.}(1996)\citenamefont
  {Bennett}, \citenamefont {Brassard}, \citenamefont {Popescu}, \citenamefont
  {Schumacher}, \citenamefont {Smolin},\ and\ \citenamefont
  {Wootters}}]{BenBraPopSchSmoWoo:J96}%
  \BibitemOpen
  \bibfield  {author} {\bibinfo {author} {\bibfnamefont {C.~H.}\ \bibnamefont
  {Bennett}}, \bibinfo {author} {\bibfnamefont {G.}~\bibnamefont {Brassard}},
  \bibinfo {author} {\bibfnamefont {S.}~\bibnamefont {Popescu}}, \bibinfo
  {author} {\bibfnamefont {B.}~\bibnamefont {Schumacher}}, \bibinfo {author}
  {\bibfnamefont {J.~A.}\ \bibnamefont {Smolin}},\ and\ \bibinfo {author}
  {\bibfnamefont {W.~K.}\ \bibnamefont {Wootters}},\ }\bibfield  {title}
  {\bibinfo {title} {Purification of noisy entanglement and faithful
  teleportation via noisy channels},\ }\href
  {https://doi.org/10.1103/PhysRevLett.76.722} {\bibfield  {journal} {\bibinfo
  {journal} {Phys. Rev. Lett.}\ }\textbf {\bibinfo {volume} {76}},\ \bibinfo
  {pages} {722} (\bibinfo {year} {1996})}\BibitemShut {NoStop}%
\bibitem [{\citenamefont {Horodecki}\ \emph {et~al.}(2009)\citenamefont
  {Horodecki}, \citenamefont {Horodecki}, \citenamefont {Horodecki},\ and\
  \citenamefont {Horodecki}}]{HorHorHorHor:J09}%
  \BibitemOpen
  \bibfield  {author} {\bibinfo {author} {\bibfnamefont {R.}~\bibnamefont
  {Horodecki}}, \bibinfo {author} {\bibfnamefont {P.}~\bibnamefont
  {Horodecki}}, \bibinfo {author} {\bibfnamefont {M.}~\bibnamefont
  {Horodecki}},\ and\ \bibinfo {author} {\bibfnamefont {K.}~\bibnamefont
  {Horodecki}},\ }\bibfield  {title} {\bibinfo {title} {Quantum entanglement},\
  }\href {https://doi.org/10.1103/RevModPhys.81.865} {\bibfield  {journal}
  {\bibinfo  {journal} {Rev. Mod. Phys.}\ }\textbf {\bibinfo {volume} {81}},\
  \bibinfo {pages} {865} (\bibinfo {year} {2009})}\BibitemShut {NoStop}%
\bibitem [{\citenamefont {Horodecki}\ and\ \citenamefont
  {Horodecki}(1996)}]{HorHor:96}%
  \BibitemOpen
  \bibfield  {author} {\bibinfo {author} {\bibfnamefont {R.}~\bibnamefont
  {Horodecki}}\ and\ \bibinfo {author} {\bibfnamefont {M.}~\bibnamefont
  {Horodecki}},\ }\bibfield  {title} {\bibinfo {title} {Information-theoretic
  aspects of inseparability of mixed states},\ }\href@noop {} {\bibfield
  {journal} {\bibinfo  {journal} {Phys. Rev. A}\ }\textbf {\bibinfo {volume}
  {54}},\ \bibinfo {pages} {1838} (\bibinfo {year} {1996})}\BibitemShut
  {NoStop}%
\bibitem [{\citenamefont {Horodecki}\ \emph {et~al.}(1997)\citenamefont
  {Horodecki}, \citenamefont {Horodecki},\ and\ \citenamefont
  {Horodecki}}]{HorHorHor:97}%
  \BibitemOpen
  \bibfield  {author} {\bibinfo {author} {\bibfnamefont {M.}~\bibnamefont
  {Horodecki}}, \bibinfo {author} {\bibfnamefont {P.}~\bibnamefont
  {Horodecki}},\ and\ \bibinfo {author} {\bibfnamefont {R.}~\bibnamefont
  {Horodecki}},\ }\bibfield  {title} {\bibinfo {title} {Inseparable two
  spin-$\frac{1}{2}$ density matrices can be distilled to a singlet form},\
  }\href@noop {} {\bibfield  {journal} {\bibinfo  {journal} {Phys. Rev. Lett.}\
  }\textbf {\bibinfo {volume} {78}},\ \bibinfo {pages} {574} (\bibinfo {year}
  {1997})}\BibitemShut {NoStop}%
\bibitem [{\citenamefont {Bobrovsky}\ \emph {et~al.}(1987)\citenamefont
  {Bobrovsky}, \citenamefont {Mayer-Wolf},\ and\ \citenamefont
  {Zakai}}]{BobMayZak:J87}%
  \BibitemOpen
  \bibfield  {author} {\bibinfo {author} {\bibfnamefont {B.~Z.}\ \bibnamefont
  {Bobrovsky}}, \bibinfo {author} {\bibfnamefont {E.}~\bibnamefont
  {Mayer-Wolf}},\ and\ \bibinfo {author} {\bibfnamefont {M.}~\bibnamefont
  {Zakai}},\ }\bibfield  {title} {\bibinfo {title} {{Some Classes of Global
  Cramer-Rao Bounds}},\ }\href {https://doi.org/10.1214/aos/1176350602}
  {\bibfield  {journal} {\bibinfo  {journal} {The Annals of Statistics}\
  }\textbf {\bibinfo {volume} {15}},\ \bibinfo {pages} {1421 } (\bibinfo {year}
  {1987})}\BibitemShut {NoStop}%
\end{thebibliography}
\end{document}